\documentclass[11pt]{amsart}

\usepackage{amsfonts,amssymb,amsmath,eucal,pinlabel,picins,array,hhline,picins}
\usepackage[all]{xy}
\usepackage{tabulary}
\usepackage{fancyhdr}

\newtheorem{theorem}{Theorem}[section]
\newtheorem*{thmintro}{Theorem}
\newtheorem{lemma}[theorem]{Lemma}
\newtheorem{proposition}[theorem]{Proposition}
\newtheorem{corollary}[theorem]{Corollary}
\theoremstyle{remark}\newtheorem{remark}[theorem]{Remark}
\theoremstyle{definition}\newtheorem{definition}[theorem]{Definition}
\theoremstyle{remark}\newtheorem{example}[theorem]{Example}

\newenvironment{romanlist}
        {\begin{enumerate}
        }
        {\end{enumerate}}

\newcounter{ticklistc}

\newcommand{\Z}{\mathbb Z}
\newcommand{\R}{\mathbb R}
\newcommand{\C}{\mathbb C}

\newcommand{\Q}{\mathcal Q}
\newcommand{\M}{\mathcal M}
\renewcommand{\S}{\mathcal S}
\renewcommand{\Re}{\mathit{Re}\,}
\renewcommand{\Im}{\mathit{Im}\,}
\newcommand{\G}{\Gamma}
\newcommand{\X}{\mathfrak X}

\newcommand{\Arf}{\mathrm{Arf}}
\newcommand{\eps}{\epsilon}
\newcommand{\St}{\mathit{St}}
\newcommand{\Si}{\Sigma}
\newcommand{\ws}{\widetilde{\Sigma}}

\begin{document}

\title{Discrete Dirac operators on Riemann surfaces and Kasteleyn matrices}

\author{David Cimasoni}   
\address{ETH Z\"urich, Departement Mathematik, R\"amistrasse 101, 8092 Z\"urich, Switzerland}
\email{david.cimasoni@math.ethz.ch}

\begin{abstract}
Let $\Sigma$ be a flat surface of genus $g$ with cone type singularities.
Given a bipartite graph $\G$ isoradially embedded in $\Si$, we define discrete analogs of the
$2^{2g}$ Dirac operators on $\Si$. These discrete objects are then shown to converge to the continuous ones, in some
appropriate sense. Finally, we obtain necessary and sufficient conditions on the pair $\G\subset\Si$ for these discrete Dirac operators to be
Kasteleyn matrices of the graph $\G$. As a consequence, if these conditions are met, the partition function of the dimer model on $\G$
can be explicitly written as an alternating sum of the determinants of these $2^{2g}$ discrete Dirac operators.
\end{abstract}

\maketitle

\pagestyle{myheadings}
\markboth{David Cimasoni}{Discrete Dirac operators and Kasteleyn matrices}

\section{Introduction}
\label{sec:intro}

A {\bf dimer covering}, or {\bf perfect matching}, of a graph $\G$ is a collection of edges with the property that each vertex is
adjacent to exactly one of these edges. Assigning weights to the edges of $\G$ allows to define a probability measure on the set of dimer coverings,
and the corresponding model is called the {\bf dimer model} on $\G$.

Dimer models are among the most studied in statistical mechanics. One of their remarkable properties is that the partition function of a dimer model on a graph $\G$
can be written as a linear combination of $2^{2g}$ Pfaffians (determinants in case of bipartite graphs), where $g$ is the genus of an orientable surface $\Si$ in which
$\G$ embeds. These $2^{2g}$ matrices,
called {\bf Kasteleyn matrices}, are skew-symmetric matrices determined by $2^{2g}$ orientations of the edges of $\G\subset\Si$, called {\bf Kasteleyn orientations}.
P. W. Kasteleyn himself proved this Pfaffian formula in the planar case \cite{Ka1,Ka2} together with the case of a square lattice embedded in the torus,
and stated the general fact \cite{Ka3}. A complete combinatorial proof of this statement was first obtained much later by Gallucio-Loebl \cite{G-L} and independently by
Tesler \cite{Tes}, who extended it to non-orientable surfaces.

In \cite{CR1}, we studied an explicit correspondance (first suggested by Kuperberg \cite{Kup}) relating spin structures on $\Si$ and Kasteleyn orientations on $\G\subset\Si$.
We also used the identification of spin structures with quadratic forms to give a geometric proof of the Pfaffian formula, together with a geometric
interpretation of its coefficients.

The partition function of free fermions on a closed Riemann surface $\Si$ of genus $g$ is also a linear combination of $2^{2g}$ determinants of Dirac operators, each term
corresponding to a spin structure on $\Si$ \cite{AMV}. Assuming that dimer models are discrete analogs of free fermions, one expects -- in addition to the known
relation between Kasteleyn orientations and spin structures -- a relation between the Kasteleyn matrix for a given Kasteleyn orientation and the Dirac operator associated
to the corresponding spin structure.

This is well understood in the planar case with the work of Kenyon \cite{Ken}. For any bipartite planar graph $\G$ satisfying some geometric condition known as
isoradiality (see below), he defined a discrete version of the Dirac operator which turns out to be closely
related to a Kasteleyn matrix of $\G$. In particular, its determinant is equal to the partition function of the dimer model on $\G$ with critical weights. In the genus one case,
the following observation was made by Ferdinand \cite{Fer} as early as 1967: For the $M\times N$ square lattice on the torus with horizontal weight $x$ and vertical
weight $y$, the determinants of the four Kasteleyn matrices behave asymptotically, in the $MN\to\infty$ limit with fixed ratio $M/N$, as a common
bulk term times the four Jacobi theta functions $\theta_k(0|\tau)$, where $\tau=i\frac{Mx}{Ny}$. This reproduces exactly the dependance of the determinant of the Dirac operators
on the different spin structures observed by Alvarez-Gaum\'e, Moore and Vafa \cite{AMV}.

The higher genus case remains somewhat mysterious. The only results available are numerical evidences, for one specific example of a square lattice embedded in a genus two surface, that the determinants of the 16 Kasteleyn matrices have a dependance that can be expressed in terms of genus two theta functions \cite{CSM}.

In short, it is fair to say that relatively little is understood of the expected relation between Kasteleyn matrices and Dirac operators on surfaces of genus $g\ge 2$.
With this paper, we aim at filling this gap. Here is a summary of our results.

We start in Section~\ref{sec:CA} by defining a discrete analog of the $\bar{\partial}$ operator on functions on a Riemann surface $\Si$. Because there is no ``canonical"
such discretization, some geometric conditions may be naturally imposed on the pair $\G\subset\Si$. Following Duffin~\cite{Duf}, Mercat~\cite{Mer}, Kenyon~\cite{Ken} and
many others, and for reasons that will become apparent along the way, we work with bipartite {\bf isoradial\/} graphs.
More precisely, we encode the complex structure on $\Si$ by a flat metric with cone type singularities supported at $S\subset\Si$, and consider locally finite graphs
$\G\subset\Si$ with bipartite structure $V(\G)=B\sqcup W$ satisfying the following conditions:
\begin{romanlist}
\item{Each edge of $\G$ is a straight line (with respect to the flat metric on $\Si$), and for some positive $\delta$, each face $f$ of $\G\subset\Si$ contains an element $x_f$
at distance $\delta$ of every vertex of $f$.}
\item{A singularity of $\Si$ is either a black vertex of $\G$, or a vertex $x_f$ of the dual graph $\G^*$.}
\end{romanlist}

For such a pair $\G\subset\Si$, we introduce a discrete $\bar{\partial}$ operator defined on $\C^B$,
and call $f\in\C^B$ {\bf discrete holomorphic\/} if $\bar{\partial}f=0$. (See Definition~\ref{def:dbar}.)
This operator satisfies some natural properties (Proposition~\ref{prop:prop})
and extends previous constructions of Mercat~\cite{Mer}, Kenyon~\cite{Ken} and Dynnikov-Novikov~\cite{D-N}.
Note that these authors impose strong conditions on $\Gamma$ in order for $\bar{\partial}$ to be defined: $\Gamma$ needs to be the double of a graph in \cite{Mer}
(and therefore, does not admit any perfect matching in higher genus cases, see Remark~\ref{rem:double}), it needs to be planar in \cite{Ken}, while only the
triangular (or dually, the hexagonal) lattice is considered in \cite{D-N}. On the other hand, our discrete
$\bar{\partial}$ operator imposes essentially no combinatorial restriction on $\G$: any locally finite bipartite graph such that each white vertex has degree
at least three can be isoradially embedded in an orientable flat surface $\Sigma$ with conical singularities $S\subset V(\Gamma^*)\cup B$ (Proposition~\ref{prop:real}).
This section is concluded with a convergence theorem:
if a sequence of discrete holomorphic functions converges to a function $f\colon\Si\to\C$ in the appropriate sense, then $f$ is holomorphic. (See Theorem~\ref{thm:CL} for
the precise statement.)

In Section~\ref{sec:Dirac}, we twist the discrete $\bar\partial$ operator on $\G\subset\Si$ by {\bf discrete spin structures\/} $\lambda$ to obtain $2^{2g}$
discrete Dirac operators
\[
D_\lambda\colon\C^B\to\C^W,
\]
provided each cone angle of $\Si$ is a multiple of $2\pi$ (Definition~\ref{def:Dirac}). The convergence theorem then takes the following form: let $\lambda_n$ be a sequence
of discrete spin structures on $\Si$ discretizing a fixed spin structure $L$. If a sequence $\psi_n$ of discrete holomorphic spinors (that is: $D_{\lambda_n}\psi_n=0$) converges
to a section $\psi$ of the line bundle $L\to\Si$, then $\psi$ is a holomorphic spinor. (See Theorem~\ref{thm:CLTS}.)

Section~\ref{sec:Kast} contains the core of this paper. First, we extend the Kasteleyn formalism from $\{\pm 1\}$-valued flat cochains on $\G\subset\Si$
(that is, Kasteleyn orientations) to $G$-valued ones for any subgroup $G$ of $\C^*$. We believe that the resulting existence statement (Proposition~\ref{prop:torsor}) and
generalized Pfaffian formula (Theorem~\ref{thm:Pf-K} and Corollary~\ref{cor:Pf}) are of independant interest. We use them to prove our main result:

\begin{thmintro}
Let $\Sigma$ be a compact oriented flat surface of genus $g$ with conical singularities supported at $S$ and cone angles multiples of $2\pi$.
Fix a graph $\Gamma$ with bipartite structure $V(\Gamma)=B\sqcup W$, isoradially embedded in $\Sigma$ so that $S\subset B\cup V(\Gamma^*)$.
For an edge $e$ of $\Gamma$, let $\nu(e)$ denote the length of the dual edge.
Finally, let $D_\lambda\colon\C^B\to\C^W$ denote the discrete Dirac operator associated to the discrete spin structure $\lambda$.

There exist $2^{2g}$ non-equivalent discrete spin structures such that the corresponding discrete Dirac operators $\{D_\lambda\}_\lambda$
give $2^{2g}$ non-equivalent Kasteleyn matrices of the weighted graph $(\Gamma,\nu)$, if and only if the following conditions hold:
\begin{romanlist}
\item{each conical singularity in $V(\G^*)$ has angle an odd multiple of $2\pi$;}
\item{for some (or equivalently, for any) choice of oriented simple closed curves $\{C_j\}$ in $\Gamma$ representing
a basis of $H_1(\Sigma;\Z)$,
\[
\sum_{b\in B\cap C_j}\alpha_b(C_j)-\sum_{w\in W\cap C_j}\alpha_w(C_j)
\]
is a multiple of $2\pi$ for all $j$, where $\alpha_v(C)$ denotes the angle made by the oriented curve $C$
at the vertex $v$ as illustrated below.}
\end{romanlist}
\begin{figure}[h]
\labellist\small\hair 2.5pt
\pinlabel {$v$} at 168 18
\pinlabel {$\alpha_v(C)$} at 175 160
\pinlabel {$C$} at 420 190
\endlabellist
\includegraphics[height=2cm]{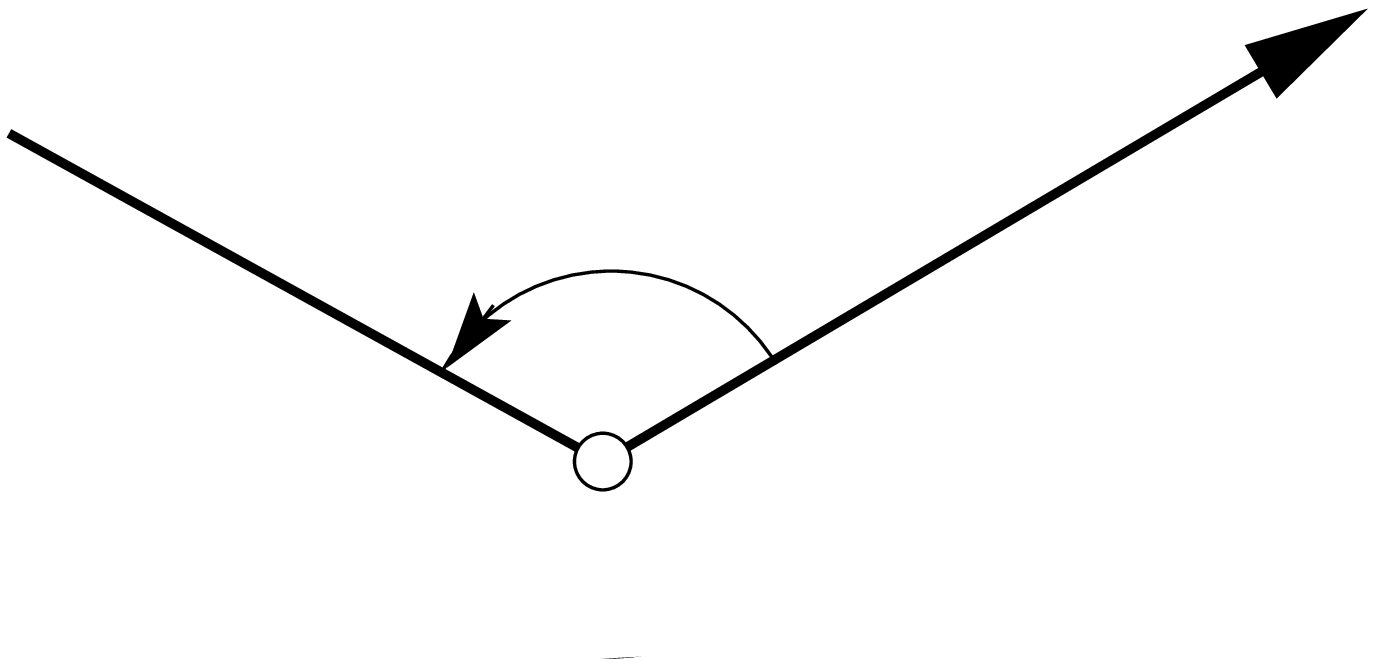}
\end{figure}
\end{thmintro}

As a consequence, given any graph $\G\subset\Si$ satisfying the conditions above, the partition function for the dimer model on $(\G,\nu)$ is given by
\[
Z(\G,\nu)=\frac{1}{2^g}\Big|\sum_{\lambda\in\S(\Sigma)}(-1)^{\mathrm{Arf}(\lambda)}\det(D_\lambda)\Big|,
\]
where $\mathrm{Arf}(\lambda)\in\Z_2$ denotes the Arf invariant of the spin structure $\lambda$ (Theorems~\ref{thm:Pf-D} and \ref{thm:Pf-Arf}).
Our final result -- Theorem~\ref{thm:real} -- states that the Dirac operators on any closed Riemann surface can be approximated by Kasteleyn matrices.
More precisely, for any closed Riemann surface of positive genus, there exist a flat surface $\Si$ with cone type singularities
inducing this complex structure, and an isoradially embedded bipartite graph $\G\subset\Si$, with arbitrarily small radius, satisfying all the hypothesis and conditions of the
theorem displayed above.

\subsection*{Acknowledgements}
The author would like to thank Mathieu Baillif, Giovanni Felder and Nicolai Reshetikhin for useful discussions.

\section{The discrete $\bar\partial$ operator on Riemann surfaces}
\label{sec:CA}

The aim of this section is to introduce a discrete analog of the $\bar{\partial}$ operator on functions on a Riemann surface, extending works of Duffin~\cite{Duf},
Mercat~\cite{Mer} and Kenyon~\cite{Ken}. As this definition requires a substantial amount of notation and terminology, we shall proceed leisurely, starting by recalling
in Paragraph~\ref{sub:fs} the main properties of flat surfaces with conical singularities. We then give in Paragraph \ref{sub:DG} discrete analogs of all the geometric
objects involved in the definition of $\bar{\partial}$ (see Table~\ref{table:dic1}). This will lead up in Paragraph \ref{sub:DDol} to the -- by then, quite natural --
definition of the discrete operator. The section is concluded with a convergence theorem (Theorem \ref{thm:CL} in Paragraph \ref{sub:CL}), justifying further
our definition.

\subsection{Flat surfaces}
\label{sub:fs}

Our discrete $\bar{\partial}$ operator will be defined for graphs embedded in so-called flat surfaces with conical singularities.
Since these objects are ubiquitous in the present paper, we devote this first paragraph to their main properties, referring to \cite{tro} for further details.

Given a positive real number $\theta$, the space
\[
C_\theta=\{(r,t)\,:\,\hbox{$r\ge 0$, $t\in\R/\theta\Z$}\}/(0,t)\sim(0,t')
\]
endowed with the metric $\mathit{ds}^2=\mathit{dr}^2+r^2\mathit{dt}^2$ is called the
{\bf standard cone of angle $\theta$\/}. Note that the cone without its tip is locally isometric to the Euclidean plane.
Let $\Sigma$ be a surface with a discrete subset $S$. A {\bf flat metric on $\Sigma$ with conical singularities} of angles
$\{\theta_x\}_{x\in S}$ supported at $S$ is an atlas $\{\phi_x\colon U_x\to U'_x\subset C_{\theta_x}\}_{x\in S}$, where $U_x$ is an open neighborhood of $x\in S$, $\phi_x$ maps
$x$ to the tip of the cone $C_{\theta_x}$, and the transition maps are Euclidean isometries.

This seemingly technical definition should not hide the fact that these objects are extremely natural: any such flat surface can be obtained by gluing polygons embedded
in $\R^2$ along pairs of sides of equal length. For example, a rectangle with opposite sides identified will define a flat torus with no singularity. On the other hand,
a regular $4g$-gon with opposite sides identified gives a flat surface of genus $g$ with a single singularity of angle $2\pi(2g-1)$. In general, the topology of the surface
is related to the cone angles by the following Gauss-Bonnet Formula: if $\Si$ is a closed flat surface with cone angles $\{\theta_x\}_{x\in S}$, then
\[
\sum_{x\in S}(2\pi-\theta_x)=2\pi\chi(\Sigma).
\]

For the purpose of this paper, the most important property of flat metrics is that they encode complex structures on oriented surfaces. Indeed, the conformal
structure on $\Sigma\setminus S$ given by a flat metric extends to the whole oriented surface, defining a complex structure on $\Sigma$. Furthermore, let $\Si$
be a closed oriented surface, $S\subset\Si$ a discrete subset, and $\{\theta_x\}_{x\in S}$ a set of positive numbers satisfying the Gauss-Bonnet Formula.
Then, for each complex structure on $\Sigma$, there exists a flat metric on $\Sigma$ with conical singularities of angles $\{\theta_x\}_{x\in S}$ supported at $S$
inducing this complex structure.

For example, any complex structure on the torus can be realized by a flat surface with no singularity: simply consider the parallelogram in the complex plane spanned by
the pair of periods of the torus, and identify the opposite sides. Similarly, deforming the regular octagon in such a way that the sides are organized into pairs
of equal length allows to realize any complex structure on the genus two surface.

\subsection{Some discrete geometry}
\label{sub:DG}

Let us begin by briefly recalling the definition of the $\partial$ and $\bar{\partial}$ operators on a Riemann surface $\Sigma$. (This will also fix some notation).
The complex structure $J$ on $\Sigma$ induces a decomposition of the complexified tangent bundle $T\Sigma^\C$ into $T\Sigma^+\oplus T\Sigma^-$, and therefore, a decomposition
of complex-valued vector fields $C^\infty(T\Sigma^\C)=C^\infty(T\Sigma^+)\oplus C^\infty(T\Sigma^-)$. The elements of $C^\infty(T\Sigma^+)$ (resp. $C^\infty(T\Sigma^-)$) are
the vector fields for which the action of $J$ is given by multiplication by $i$ (resp. $-i$). Similarly, the complex cotangent bundle splits, resulting in a decomposition
of the complex-valued 1-forms on $\Sigma$:
\[
\Omega^1(\Sigma,\C)=C^\infty(T^*\Sigma^\C)=C^\infty(T^*\Sigma^+)\oplus C^\infty(T^*\Sigma^-)=\Omega^{1,0}(\Sigma)\oplus\Omega^{0,1}(\Sigma).
\]
Note that the forms of type $(1,0)$ (resp. $(0,1)$) are the 1-forms $\varphi$ such that for any vector field $V$, $\varphi(J(V))$ is equal to $i\varphi(V)$
(resp. $-i\varphi(V)$). Finally, the exterior derivative $d\colon C^\infty(\Sigma)\to\Omega^1(\Sigma)$ induces a $\C$-linear map
$d^\C\colon C^\infty(\Sigma,\C)\to\Omega^1(\Sigma,\C)=\Omega^{1,0}(\Sigma)\oplus\Omega^{0,1}(\Sigma)$, whose composition with the natural projections
defines the Dolbeault operators $\partial\colon C^\infty(\Sigma,\C)\to\Omega^{1,0}(\Sigma)$ and $\bar\partial\colon C^\infty(\Sigma,\C)\to\Omega^{0,1}(\Sigma)$.
Recall that a function $f\in C^\infty(\Sigma,\C)$ is holomorphic if and only if $\bar\partial f$ is zero.

\medskip

We are now ready to start our discretization procedure. First and foremost, a Riemann surface $\Sigma$ is a surface. To encode the topology of $\Sigma$,
fix a locally finite graph $\Gamma\subset\Sigma$ with vertex set $V(\Gamma)$ and edge set $E(\Gamma)$, such that $\Sigma\setminus\Gamma$ consists of disjoint open discs.
In other words, $\Gamma$ is the 1-skeleton of a cellular decomposition of $\Sigma$.
For notational simplicity, we shall assume throughout this section that $\Gamma$ has neither multiple edges, nor valency one vertices.
(Note however that all our results hold in the general case as well.)

As explained in the previous paragraph, a standard and beautiful way to encode a complex structure on an oriented surface $\Sigma$ is to endow this surface with a
flat metric with conical singularities. Note that any point in $\Sigma\setminus S$ has a well-defined tangent space. In particular, the space $\X(\Sigma)$ of vector fields
on $\Sigma$ can be naively discretized by $\X(D)$, the space of vector fields along some discrete subset $D\subset\Sigma\setminus S$.

\begin{table}[t]\centering
\caption{Discretization dictionary, part 1}
\label{table:dic1}

\setlength\extrarowheight{4pt}
\setlength\arrayrulewidth{.8pt}

\begin{tabulary}{\textwidth}{||C||C||}\hhline{|t:=:t:=:t|}

\em{the geometric object} & \em{the discrete analog} \\ \hhline{|:=::=:|}

a surface $\Sigma$ & a graph $\Gamma\subset\Sigma$ inducing a cellular decomposition of $\Sigma$ \\ \hhline{||-||-||}

a conformal structure on $\Sigma$ & a flat metric on $\Sigma$ with conical singularities $S\subset\Sigma$ \\ \hhline{||-||-||}

the space $\X(\Sigma)$ of vector fields on $\Sigma$ & the space $\X(D)$ of vector fields along some discrete subset $D\subset\Sigma$ \\ \hhline{||-||-||}

the decomposition $C^\infty(T\Sigma^\C)=C^\infty(T\Sigma^+)\oplus C^\infty(T\Sigma^-)$ induced by an almost complex structure on $\Sigma$ &
a bipartite structure $V(\Gamma)=B\sqcup W$ together with a perfect matching $M$ on the graph $\Gamma$, inducing a decomposition $\X(D_M)^\C=\X(B)\oplus\X(W)$\\\hhline{||-||-||}

the space $\Omega^{1,0}(\Sigma)$ of $(1,0)$-forms on $\Sigma$ & $\Omega^1(B)=\prod_{b\in B}\X(b)^*$ \\ \hhline{||-||-||}

the space $\Omega^{0,1}(\Sigma)$ of $(0,1)$-forms on $\Sigma$ & $\Omega^1(W)=\prod_{w\in W}\X(w)^*$ \\ \hhline{||-||-||}

the space $C^\infty(\Sigma,\C)$ of complex functions on $\Sigma$ & $\C^B\simeq\C^W$, identified via $M$ \\ \hhline{||-||-||}

$\iint_{P}\bar\partial F\,dx\,dy=-\frac{i}{2}\int_{\partial P} F\,dz$ & the definition of the discrete $\bar\partial$ operator $\bar\partial\colon\C^B\to \Omega^1(W)$\\

\hhline{|b:=:b:=:b|}

\end{tabulary}
\end{table}

Next, we wish to encode in the graph $\Gamma$ the decomposition of $\X(\Sigma)^\C=C^\infty(T\Sigma^\C)$ induced by the almost complex structure.
A convenient way to do so is to consider a decomposition $V(\Gamma)=B\sqcup W$ of the vertices of $\Gamma$ into, say, black and
white vertices, together with a perfect matching $M$ on $\Gamma$ pairing each black vertex $b\in B$ with a white one $w\in W$ and vice versa.
Any perfect matching will do the job, provided
$\Gamma$ is {\bf bipartite}\/, that is: no edge of $\Gamma$ links two vertices of the same color. Hence, when the surface $\Sigma$ is endowed with an almost complex structure,
it is natural to consider a bipartite graph $\Gamma\subset\Sigma$ together with a perfect matching $M$ on it. The discrete analog of the decomposition
$C^\infty(T\Sigma^\C)=C^\infty(T\Sigma^+)\oplus C^\infty(T\Sigma^-)$ is then given by $\X(D_M)^\C=\X(B)\oplus\X(W)$, where $D_M\subset\Sigma$
denotes the discrete set consisting of the middle point of each edge in $M$. The complex structure on $\X(B)$ (resp. $\X(W)$) is such that multiplication by $i$ corresponds
to the 90 degrees rotation of the tangent vectors in the positive (resp. negative) direction, which we will represent in our figures as counterclockwise (resp. clockwise).

In the same way, $\Omega^{1,0}(\Sigma)$ will be encoded by the space $\Omega^1(B):=\prod_{b\in B}\X(b)^*$ and $\Omega^{0,1}(\Sigma)$ by $\Omega^1(W):=\prod_{w\in W}\X(w)^*$,
where $\X(v)^*$ denotes the dual to the 1-dimensional complex vector space $\X(v)$ of tangent vectors at the vertex $v\in V(\Gamma)$. Finally, the space
$C^\infty(\Sigma,\C)$ can be discretized both by $\C^B$ and by $\C^W$, which are identified via the perfect matching $M$.

Table \ref{table:dic1} summarizes our notations and the dictionary between the geometric and discrete objects considered in this paragraph. (The last entry will be explained
shortly.)

\subsection{The discrete $\bar\partial$ operator}
\label{sub:DDol}

Let $\Gamma$ be a bipartite graph embedded in a flat surface $\Sigma$ with conical singularities supported at $S$.
According to the discussion above, the operator $\bar\partial\colon C^\infty(\Sigma,\C)\to\Omega^{0,1}(\Sigma)$ should discretize to a $\C$-linear map
\[
\bar{\partial}\colon \C^B\to\Omega^1(W)=\prod_{w\in W}\X(w)^*.
\]
These spaces make sense as soon as no white vertex is a singularity, so let us only assume $S\subset\Sigma\setminus W$ for now.

Following the terminology of Kenyon \cite{Ken}, we shall say that $\Gamma$ is {\bf isoradially embedded\/} in $\Sigma$ if each edge of $\Gamma$ is a straight line, and if for some $\delta>0$, each face $f$ of $\Gamma\subset\Sigma$ contains an element $x_f$ such that $d(x_f,v)=\delta$ for all vertices $v$ of $\partial f$. We shall furthermore assume
that a singularity of $\Sigma$ is either a black vertex $b$ of $\Gamma$, or a vertex $x_f$ of the dual graph $\Gamma^*$, that is: $S\subset V(\Gamma^*)\cup B$.

\begin{figure}[h]
\centerline{\psfig{file=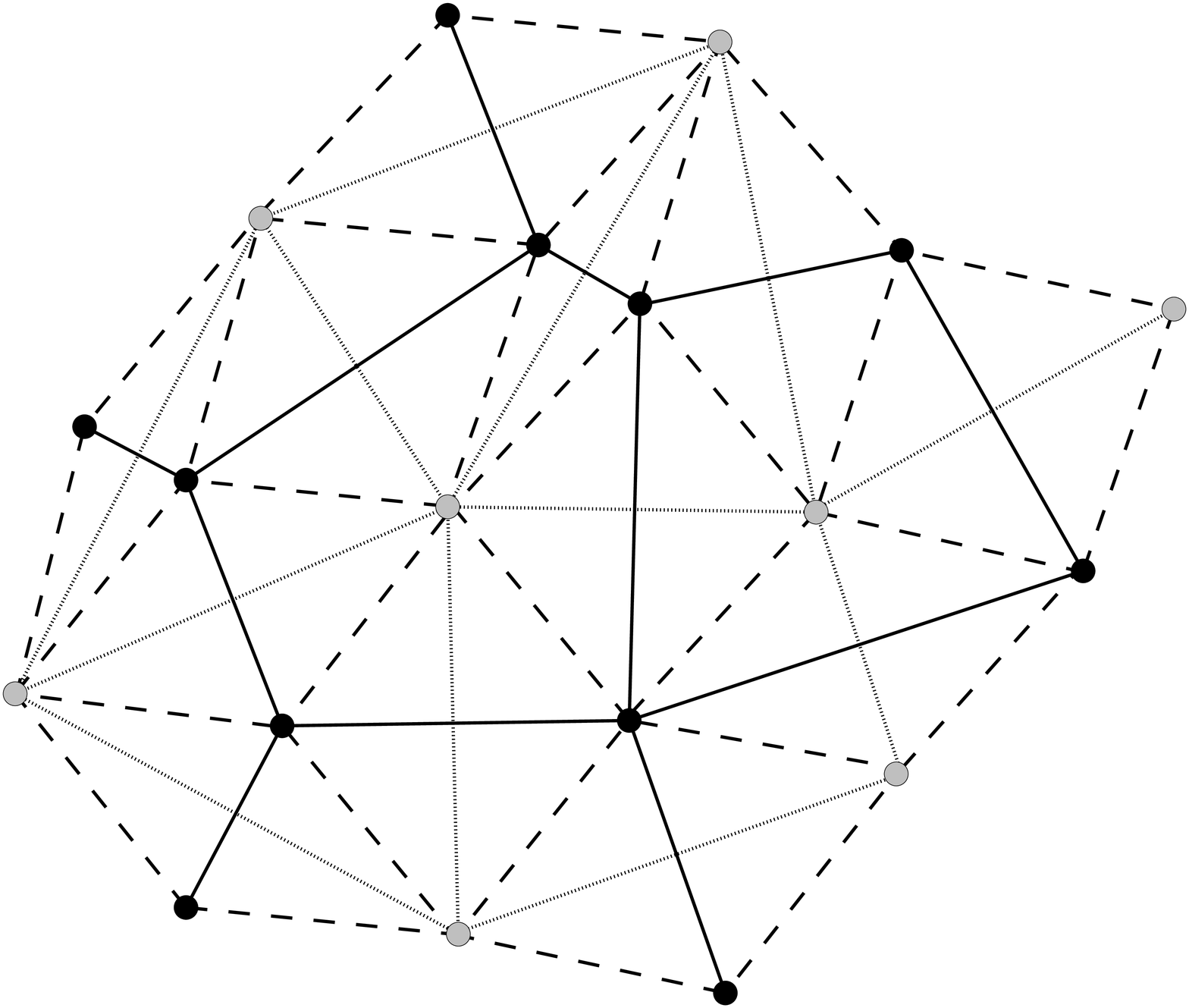,height=6cm}}
\caption{A isoradial graph $\Gamma$ (black vertices, solid edges), its dual graph $\Gamma^*$ (lighter vertices and edges), and the
associated rhombic lattice (all vertices, dashed edges).}
\label{fig:rhombus}
\end{figure}

Given an isoradially embedded graph $\Gamma\subset\Sigma$, the associated {\bf rhombic lattice\/} is the graph $R_\G$ with vertex set $V(\Gamma)\cup V(\Gamma^*)$ and edges
joining each vertex of $\Gamma$ with the center of the adjacent faces, as illustrated in Figure~\ref{fig:rhombus}. Since $\G$ induces a cellular decomposition of $\Sigma$,
so does the rhombic lattice $R_\G\subset\Sigma$. Furthermore, as the singularities of $\Sigma$ lie among the vertices of the rhombic lattice,
one easily checks that the faces of this lattice are actual planar rhombi.
Therefore, the metric space $\Sigma$ should be understood as planar (paper) rhombi pasted together along their boundary edges.

For a fixed white vertex $w\in W$, let $\mathit{St}(w)\subset\Sigma$ denote the {\bf star of $w$\/} in the rhombic lattice, that is, the union of all the closed rhombi
adjacent to $w$ (see Figure~\ref{fig:star}). As $w$ does not belong to the singular set either,
the whole star can be isometrically embedded in the Euclidean plane, as ``demonstrated" by cutting and pasting paper rhombi. Hence, one should really think of $\Sigma$ as planar stars as in Figure~\ref{fig:star} pasted along their edges. (Note that since multiple
edges and valency one vertices are not allowed, we avoid the case where boundary edges of a star are glued together. But again, the difficulty of the general case would
only be notational.)

\begin{figure}[h]
\labellist\small\hair 2.5pt
\pinlabel {$b_1$} at 870 170
\pinlabel {$b_2$} at 710 540
\pinlabel {$b_3$} at 280 670
\pinlabel {$b_4$} at -30 290
\pinlabel {$b_m$} at 380 -20
\pinlabel {$x_1$} at 660 295
\pinlabel {$x_2$} at 458 505
\pinlabel {$x_3$} at 168 434
\pinlabel {$x_4$} at 130 80
\pinlabel {$x_m$} at 615 70
\pinlabel {$w$} at 345 180
\endlabellist
\centerline{\psfig{file=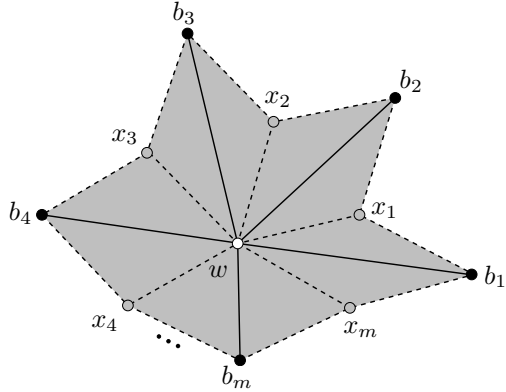,height=4.5cm}}
\caption{The star $\mathit{St}(w)$ of the white vertex $w$.}
\label{fig:star}
\end{figure}

As mentioned in the introduction, isoradial graphs first appeared in the work of Duffin \cite{Duf} (in the form of planar rhombic lattices) as a large class of graphs for which
the Cauchy-Riemann operator admits a nice discretization. It should therefore not come as a surprise that such a condition is necessary for
our definition.

\begin{definition}\label{def:dbar}
Let $\Gamma$ be a bipartite graph isoradially embedded in a flat surface $\Sigma$ with conical singularities $S\subset V(\Gamma^*)\cup B$.
Given an edge $(w,b)$ of $\Gamma$, let $\nu(w,b)$ denote the length of the dual edge.
The {\bf discrete $\bar{\partial}$ operator\/} is the linear map $\bar{\partial}\colon \C^B\to\Omega^1(W)=\prod_{w\in W}\X(w)^*$ defined by
\[
(\bar{\partial}f)(w)(V_w)=: (\bar{\partial}_wf)(V_w)=\frac{|V_w|}{2\mathit{Area}(\St(w))}\sum_{b\sim w}\nu(w,b)e^{i\vartheta_V(w,b)}\,f(b)
\]
for $f\in\C^B$, $w\in W$ and $V_w\in\X(w)$.
The sum is over all vertices $b$ adjacent to $w$, and $\vartheta_V(w,b)$
denotes the angle at $w\in W$ from the tangent vector $V_w$ to the oriented edge $(w,b)$, as illustrated in Figure~\ref{fig:theta}.

A function $f\in\C^B$ is {\bf discrete holomorphic\/} if $\bar{\partial}f=0$.
\end{definition}

\begin{figure}[h]
\labellist\small\hair 2.5pt
\pinlabel {$w$} at 20 85
\pinlabel {$b$} at 200 290
\pinlabel {$\vartheta_V(w,b)$} at 185 85
\pinlabel {$V_w$} at 245 -5
\endlabellist
\centerline{\psfig{file=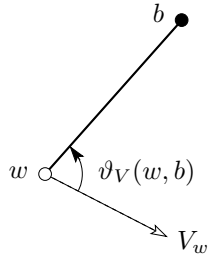,height=3cm}}
\caption{Definition of the angle $\vartheta_V(w,b)$.}
\label{fig:theta}
\end{figure}

This definition should be understood as a discretization of the formula
\[
\iint_{P}\bar\partial F(x+iy)\,dx\,dy=-\frac{i}{2}\int_{\partial P} F(z)\,dz.
\]
Indeed, given $f\in\C^B$, let $\widehat{f}\colon\Sigma\to\C$ be defined (almost everywhere) by $\widehat{f}(p)=f(b)$ if $p$ belongs to the interior of the star $\St(b)$.
Fix a white vertex $w$, and consider an isometric embedding $\phi$ of the corresponding star $\St(w)$ into $\C$. Setting $V_w=(T_w\phi)^{-1}(1)$, $F=\widehat{f}\circ\phi^{-1}$,
and using the notations of Figure~\ref{fig:star}, we get
\begin{eqnarray*}
(\bar\partial_w f)(V_w)&\approx&(\bar\partial F)(\phi(w))\\
&\approx&\frac{1}{\mathit{Area}(\St(w))}\iint_{\phi(\St(w))}\bar\partial F(x+iy)\,dx\,dy\\
&=&\frac{-i}{2\mathit{Area}(\St(w))}\int_{\phi(\partial\St(w))} F(z)\,dz\\
&=&\frac{1}{2\mathit{Area}(\St(w))}\sum_{j=1}^mi(\phi(x_{j-1})-\phi(x_j))f(b_j)\\
&=&\frac{1}{2\mathit{Area}(\St(w))}\sum_{b\sim w}\nu(w,b)e^{i\vartheta_V(w,b)}\,f(b).
\end{eqnarray*}

This definition extends previous work of Duffin~\cite{Duf}, Mercat~\cite{Mer}, Kenyon~\cite{Ken}, Dynnikov-Novikov~\cite{D-N} and Chelkak-Smirnov~\cite{C-S}, as described below.

\noindent{\bf Special case 1.} If the flat surface $\Sigma$ has trivial holonomy, there is a well-defined constant vector field in $\X(W)$.
Evaluating the discrete $\bar{\partial}$ operator at this vector field yields a map $K\colon\C^B\to\C^W$.
In the special case where $\Sigma$ has no singularity (which is only possible if $\Sigma$ is the plane, a cylinder or a torus),
this map coincides with the discrete $\bar{\partial}$ operator on planar bipartite isoradial graphs defined by Kenyon \cite{Ken} (up to the normalization constant).
For the planar hexagonal lattice, this operator is conjugate to the discrete $\bar{\partial}$ operator considered by Dynnikov-Novikov in \cite{D-N}.
\smallskip

\noindent{\bf Special case 2.}
Let $\mathcal{G}$ be a (non-necessarily bipartite) graph embedded in a flat surface $\Sigma$ with singularities $S\subset V(\mathcal{G})\cup V(\mathcal{G}^*)=:\Lambda$.
The {\bf double of $\mathcal{G}$\/} is the bipartite graph
$\Gamma=\mathcal{G}\cup\mathcal{G}^*\subset\Sigma$ with black vertices $B=\Lambda$ and
white vertices $W=E(\mathcal{G})\cap E(\mathcal{G}^*)=:\Diamond$, or the other way around. If $\mathcal{G}$ is isoradially embedded in $\Sigma$, then
so is $\mathcal{G}^*$ (with same radius $\delta$) and $\Gamma$ (with radius $\delta/2$).

In the special case of planar double graphs, an element $f\in\C^B$ is either a function on the vertices of the associated rhombic lattice
(when $B=\Lambda$), or a function on the set of rhombi (when $B=\Diamond$).
As $\Sigma=\C$, our discrete $\bar{\partial}$ operator can be evaluated at the constant vector field $1\in T_p\C=\C$, yielding two
maps $\C^\Lambda\to\C^\Diamond$ and $\C^\Diamond\to\C^\Lambda$. These correspond exactly to the two discrete $\bar\partial$ operators defined
by Chelkak and Smirnov in \cite{C-S}.

\smallskip

\parpic[r]{$\begin{array}{c}
\labellist\small\hair 2.5pt
\pinlabel {$w$} at 260 90
\pinlabel {$x$} at 0 80
\pinlabel {$x'$} at 450 80
\pinlabel {$y$} at 200 -10
\pinlabel {$y'$} at 200 260
\endlabellist
\includegraphics[height=1.5cm]{CR}
\end{array}$}
\noindent{\bf Special case 3.}
Finally, let us consider the more general case of double graphs isoradially embedded in flat surfaces with conical singularities, with bipartite structure $B=\Lambda$.
Let $w\in W=\Diamond$ be a fixed rhombus, and let $x,y,x',y'$ denote its vertices enumerated counterclockwise, as illustrated above.
Then, the equality $\bar\partial_w f=0$ coincides with the very intuitive ``discrete Cauchy-Riemann equation" studied by Duffin~\cite{Duf} (in the planar case)
and Mercat~\cite{Mer}:
\[
\frac{f(y')-f(y)}{d(y,y')}=i\,\frac{f(x')-f(x)}{d(x,x')}.
\]

\medskip

Let us mention several natural properties of the discrete $\bar\partial$ operator.

\begin{proposition}\label{prop:prop}
The discrete $\bar\partial$ operator satisfies the following properties.
\begin{romanlist}
\item{If $f\in\C^B$ is constant, then $\bar\partial f=0$.}
\item{Given a fixed white vertex $w$, let $f\in\C^{B\cap\mathit{St}(w)}$ be the restriction of a coordinate chart on a neighborhood of $\mathit{St}(w)$.
Then, $\bar\partial_w f=0$ .}
\item{Let $f\in\C^{B\cap\mathit{St}(w)}$ be as in {\it (ii)\/} above. Then its complex conjugate $\bar f$ satisfies $(\bar{\partial}_w\bar f)(V_w)=|V_w|$.}
\end{romanlist}
\end{proposition}
\begin{proof}
Let $f\in\C^B$ be a constant function. Fix a white vertex $w\in W$ and a unit tangent vector $V_w\in\X(w)$. Let $\phi\colon\mathit{St}(w)\to\C$
be an isometric embedding mapping $w$ to the origin and $V_w$ to the direction of the unit vector $1\in\C=T_0\C$.
With the notation of Figure~\ref{fig:star}, observe that for all $j$,
\[
\nu(w,b_j)e^{i\vartheta_V(w,b_j)}=i(\phi(x_j)-\phi(x_{j-1})),
\]
as both these complex numbers have same modulus and argument. It follows that
\[
(\bar{\partial}_wf)(V_w)=\frac{f(b_1)\,i}{2\mathit{Area}(\St(w))}\sum_{j=1}^m\left(\phi(x_{j-1})-\phi(x_j)\right)=0,
\]
proving the first claim.
To check the second one, let $\phi\colon U\to\C$ be a coordinate chart with $\St(w)\subset U$, and let $V_w\in\X(w)$ be the unit tangent vector
corresponding to the edge $(w,x_m)$ (recall Figure~\ref{fig:star}). By the first point above, it may be assumed that $\phi(w)=0$. Clearly, one can also assume that
$T_w\phi(V_w)=1\in T_0\C$. For $j=1,\dots m$, let $\alpha_j$ denote the angle at $w$ of the rhombus corresponding to the edge $(w,b_j)$. Note that
\[
\sin(\alpha_j)=2\sin(\alpha_j/2)\cos(\alpha_j/2)=\frac{\nu(w,b_j)d(w,b_j)}{2\delta^2}.
\]
Since $f(b_j)=\phi(b_j)=d(w,b_j)e^{i\vartheta_V(w,b_j)}$, we get
\begin{eqnarray*}
2\mathit{Area}(\St(w))(\bar\partial_w f)(V_w)&=&\sum_{j=1}^m\nu(w,b_j)e^{i\vartheta_V(w,b_j)}f(b_j)\\
&=&\sum_{j=1}^m\nu(w,b_j)d(w,b_j)e^{i2\vartheta_V(w,b_j)}\\
&=&-i\delta^2\sum_{j=1}^m(e^{i\alpha_j}-e^{-i\alpha_j})e^{i(\sum_{k=1}^{j-1}2\alpha_k+\alpha_j)}\\
&=&-i\delta^2\sum_{j=1}^m\left(e^{2i\sum_{k=1}^{j}\alpha_k}-e^{2i\sum_{k=1}^{j-1}\alpha_k}\right)=0, \\
\end{eqnarray*}
using the fact that $\sum_{k=1}^m\alpha_k=2\pi$. Finally, 
\[
(\bar\partial_w\bar f)(V_w)=\frac{|V_w|}{2\mathit{Area}(\St(w))}\sum_{j=1}^m\nu(w,b_j)d(w,b_j)=|V_w|,
\]
showing the third claim.
\end{proof}

\begin{remark}\label{rem:double}
As pointed out in the special cases above, most authors have considered discrete $\bar\partial$ operators defined on double graphs only.
It is however crucial for us to consider more general graphs, for the following reason.
In Section~\ref{sec:Kast}, we shall turn to the problem of counting perfect matchings on a (finite) bipartite graph $\Gamma$ embedded in a (compact) surface $\Sigma$.
For such a matching to exist, one necessary condition is that the number of black vertices equals the number of white ones. But in the case of a double graph
$\Gamma={\mathcal D}({\mathcal G})\subset\Sigma$, this condition gives
\[
0=|B|-|W|=|V({\mathcal G})|+|F({\mathcal G})|-|E({\mathcal G})|=\chi(\Sigma).
\]
Hence, no double graph as above admits a perfect matching unless $\Sigma$ is a torus.
\end{remark}

On the other hand, our setting imposes almost no restriction on the combinatorial type of the graphs considered:

\begin{proposition}\label{prop:real}
Let $\Gamma$ be a locally finite bipartite graph such that each white vertex has degree at least three. Then, $\Gamma$ can be isoradially embedded in an orientable flat
surface $\Sigma$ with conical singularities $S\subset V(\Gamma^*)\cup B$.
\end{proposition}
\begin{proof}
For each $w\in W$, fix a cyclic ordering of the $m$ adjacent edges (so that multiple edges are consecutive) and form the symmetric star $\St(w)$ by pasting together
$m$ rhombi of side length $\delta$ and of angle $2\pi/m$ according to this ordering. Note that the surface $\St(w)$ is endowed with an orientation given by the cyclic
ordering. In case of multiple edges or black vertices of degree 1, identify the corresponding boundary edges of $\St(w)$ accordingly. (This respects the orientation of the star.)
For each $b\in B$, fix a cyclic ordering of the adjacent edges, and glue the stars $\St(w)$ along their boundary edges according to these orderings, in the unique
way compatible with the orientations of the stars. The result is an oriented flat surface $\Sigma$ with conical singularities supported at $S$.
By construction, $\Gamma$ is isoradially embedded in $\Sigma$ and $S$ is contained in $V(\Gamma^*)\cup B$.
\end{proof}

\parpic[r]{$\begin{array}{c}
\includegraphics[height=2.3cm]{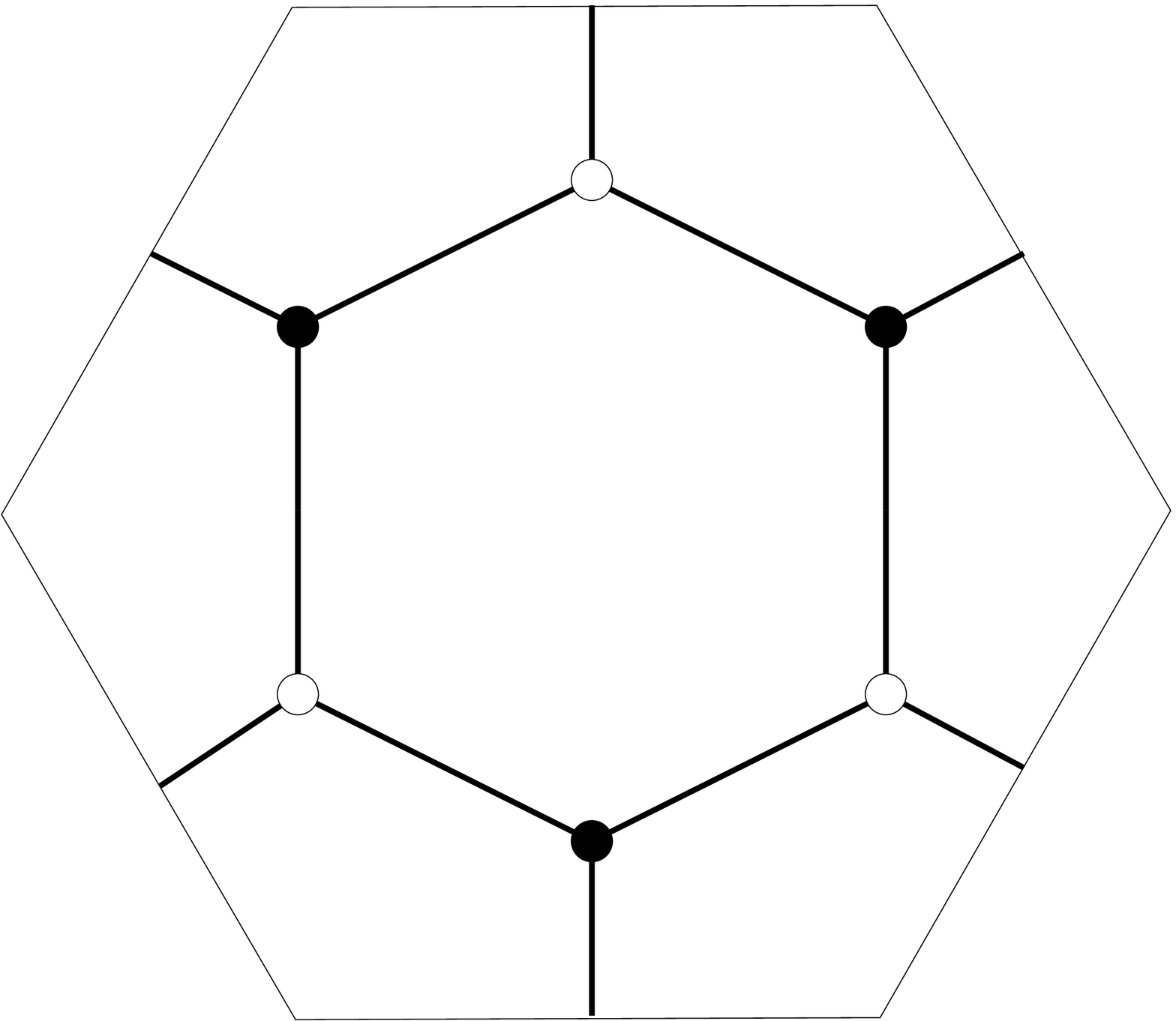}
\end{array}$}
As an example, consider the complete bipartite graph $K_{3,3}$. With a natural choice of the cyclic orderings around the vertices, the construction above yields
the honeycomb lattice embedded in the flat torus illustrated opposite. (The pairs of opposite sides of the big hexagon are identified.)

\medskip

To conclude this paragraph, note that if $S\subset \Sigma\setminus B$, one can define the {\bf discrete $\partial$ operator\/} as the $\C$-linear map
$\partial\colon \C^W\to\Omega^1(B)=\prod_{b\in B}\X(b)^*$ given by
\[
(\partial_bg)(U_b)=\frac{|U_b|}{2\mathit{Area}(\St(b))}\sum_{w\sim b}\nu(w,b)e^{-i\vartheta_U(w,b)}\,g(w)
\]
for $g\in\C^W$, $b\in B$ and $U_b\in\X(b)$. Here again, the sum is over all vertices $w$ adjacent to $b$, and
$\vartheta_U(w,b)$ denotes the angle at $b\in B$ from the tangent vector $U_b$ to the oriented edge $(w,b)$.
This construction generalises the one given in \cite{Ken}, which corresponds to the case with no singularity.
However, the discrete $\bar{\partial}$ operator being sufficient for our purposes, we shall not study $\partial$ in the
present paper.

\subsection{A convergence theorem}
\label{sub:CL}

The aim of this paragraph is to prove the following result.

\begin{theorem}
\label{thm:CL}
Let $\Sigma$ be a flat surface with conical singularities supported at $S$.
Consider a sequence $\Gamma_n$ of bipartite graphs isoradially embedded in $\Sigma$ with $S\subset V(\Gamma_n^*)\cup B_n$.
Assume that the radii $\delta_n$ of $\G_n$ converge to $0$, and that there is some $\eta>0$ such that
all rhombi angles of all these $\Gamma_n$'s belong to $[\eta,\pi-\eta]$.
Let $f_n\in\C^{B_n}$ be a sequence of discrete holomorphic functions converging to a function $f\colon\Sigma\to\C$
in the following sense: for any sequence $x_n\in B_n$ converging in $\Sigma$, the sequence
$f_n(x_n)$ converges to $f(\lim_nx_n)$ in $\C$. Then, the function $f$ is holomorphic in $\Sigma$.
\end{theorem}

Our proof will follow the same lines as the one of Mercat~\cite[pp.192-195]{Mer}, a notable exception being the discrete Morera Theorem below (Lemma~\ref{lemma:Morera}).
Let us start with a straightforward generalization of \cite[Lemma 2]{Mer}.

\begin{lemma}
\label{lemma:conv}
Let $X$ be a metric space. Consider a sequence of functions $f_n\colon X\to\C$ converging to
$f\colon X\to\C$ in the following sense: for any convergent sequence $x_n$ in $X$, the sequence
$f_n(x_n)$ converges to $f(\lim_nx_n)$ in $\C$. Then, the function $f$ is continuous, and is the uniform limit
of $f_n$ on any compact.
\end{lemma}
\begin{proof}
To show that $f$ is continuous at an arbitrary point $x\in X$, pick a sequence $x_j$ converging to $x$ in $X$.
For any $j=1,2,\dots$, the hypothesis applied to the constant sequence $x_j$ yields the existence of an index $n_j$ such that $|f_{n_j}(x_j)-f(x_j)|<1/j$.
Let $y_n$ be the sequence given by $y_n=x_j$ if $n=n_j$, and
$y_n=x$ else. As $y_n$ converges to $x$, $f_n(y_n)$ converges to $f(x)$ and so does the subsequence $f_{n_j}(x_j)$. It follows that
\[
|f(x_j)-f(x)|\le |f(x_j)-f_{n_j}(x_j)|+|f_{n_j}(x_j)-f(x)|
\]
is arbitarily small, proving the first claim.

To show the second one, let us assume by contradiction that $f_n$ does not converge uniformly on some fixed compact $C\subset X$.
This would imply the existence of a convergent sequence $x_n$ in $C$ with $|f_n(x_n)-f(x_n)|$ greater than some $\varepsilon>0$ for all $n$.
On the other hand, the hypothesis together with the continuity of $f$ at $x=\lim_nx_n$ imply
\[
|f_n(x_n)-f(x_n)|\le |f_n(x_n)-f(x)|+|f(x)-f(x_n)|<\varepsilon
\]
for $n$ big enough, a contradiction.
\end{proof}

\begin{lemma}
\label{lemma:bound}
Let $\Gamma$ be a graph isoradially embedded in the Euclidean plane, such that all rhombi angles belong to the interval $[\eta,\pi-\eta]$ for some $\eta>0$.
Then, for any vertex $v$ of $\Gamma$, and for any two elements $x,x'$ in the boundary $\partial\St(v)$ of the star of $v$,
\[
\frac{d_{\partial\St(v)}(x,x')}{|x-x'|}\le\frac{2\pi}{\eta\sin(\eta/2)}.
\]
\end{lemma}
\begin{proof}
To simplify the notation, let $S$ stand for the star $\St(v)$ throughout this proof.
If $x$ and $x'$ belong to the same rhombus of $S$, then the quotient above is easily seen to be maximal when $x$ and $x'$ are at the same distance of the vertex opposite
to $v$. In such a case, this quotient is equal to $1/\sin(\alpha/2)$, where $\alpha$ denotes the angle of this rhombus at $v$. Since $\eta\le\alpha\le\pi-\eta$,
it follows
\[
\frac{d_{\partial S}(x,x')}{|x-x'|}\le\frac{1}{\sin(\alpha/2)}\le\frac{1}{\sin(\eta/2)}\le\frac{2\pi}{\eta\sin(\eta/2)}.
\]
Let us now assume that $x$ and $x'$ belong to adjacent edges, but distinct rhombi of $S$. If the corresponding angles at $w$ are equal to $\alpha$ and $\alpha'$,
then the argument above gives the inequality
\[
\frac{d_{\partial S}(x,x')}{|x-x'|}\le\frac{1}{\sin((\alpha+\alpha')/2)}\le\frac{1}{\sin(\eta)}\le\frac{2\pi}{\eta\sin(\eta/2)}.
\]
If $x$ and $x'$ lie on adjacent rhombi, but non-adjacent edges of the star, then $|x-x'|$ is bounded below by $\delta\sin(\eta)$, where $\delta$ denotes the length of
the rhombus edges. On the other hand, $d_{\partial S}(x,x')\le 4\delta$ as $x$ and $x'$ belong to adjacent rhombi. Therefore, in this case
\[
\frac{d_{\partial S}(x,x')}{|x-x'|}\le\frac{4}{\sin(\eta)}\le\frac{2\pi}{\eta\sin(\eta/2)}.
\]
Finally, consider the case where $x$ and $x'$ do not belong to adjacent rhombi. Fix a rhombus between them, and let $\alpha$ denote its angle at $v$. This time, 
$|x-x'|$ is bounded below by $2\delta\sin(\alpha/2)\ge 2\delta\sin(\eta/2)$, while the distance in $\partial S$ is bounded above by
half of the length of $\partial S$, that is
\[
d_{\partial S}(x,x')\le\frac{\ell(\partial S)}{2}=\#\{\text{rhombi in $S$}\}\cdot\delta\le\frac{2\pi\delta}{\eta}.
\]
This implies
\[
\frac{d_{\partial S}(x,x')}{|x-x'|}\le\frac{\pi}{\eta\sin(\eta/2)}\le\frac{2\pi}{\eta\sin(\eta/2)},
\]
and concludes the proof.
\end{proof}

The last lemma requires some preliminaries.
As above, let $\Gamma$ be a bipartite graph isoradially embedded in a flat surface $\Sigma$.
Given a function $\widehat{f}\colon\Sigma\to\C$ and a white vertex $w$ of $\Gamma$, set
\[
\int_{\partial\St(w)}\widehat{f}:=\int_a^b\widehat{f}(\gamma(t))(\phi\circ\gamma)'(t)\,dt,
\]
where $\gamma\colon[a,b]\to\St(w)$ is a parametrization of $\partial\St(w)$ and $\phi\colon\St(w)\hookrightarrow\C$ is an isometric embedding.
Obviously, the value of this integral depends on the choice of the chart $\phi$. However, the choice of another chart would multiply the result by a modulus 1 complex number.
In particular, the vanishing of this integral does not depend on such a choice, and the following statement makes sense.

\begin{lemma}[discrete Morera Theorem]
\label{lemma:Morera}
Let $\Gamma$ be a bipartite graph isoradially embedded in a flat surface $\Sigma$
with conical singularities $S\subset V(\Gamma^*)\cup B$.
Given $f\in\C^B$, let $\widehat{f}\colon\Sigma\to\C$ be the function defined by
$\widehat{f}(p)=\frac{1}{m}\sum_{j=1}^m f(b_j)$ if $p$ belongs to $\bigcap_{j=1}^m\St(b_j)$.
Then, $f$ is discrete holomorphic if and only if $\int_{\partial\St(w)}\widehat{f}=0$ for all $w\in W$.
\end{lemma}
\begin{proof}
Let $w$ be a white vertex, and let $\phi\colon\St(w)\to\C$ be an isometric embedding of the corresponding star. Fixing a unit vector $V_w\in\X(w)$,
one can assume that $T_w\phi$ maps $V_w$ to $1\in\C$. With the notation of Figure~\ref{fig:star}, we get the equality
\[
\int_{\partial\St(w)}\widehat{f}=\sum_{j=1}^m(\phi(x_j)-\phi(x_{j-1}))f(b_j)=2i\mathit{Area}(\St(w))(\bar{\partial}_wf)(V_w),
\]
and the lemma follows.
\end{proof}

\begin{proof}[Proof of Theorem~\ref{thm:CL}]
As in Lemma~\ref{lemma:Morera}, extend $f_n\in\C^{B_n}$ to a function $\widehat{f}_n\colon\Sigma\to\C$ by setting
$\widehat{f}_n(p)=\frac{1}{m}\sum_{j=1}^m f_n(b_j)$ if $p$ belongs to $\bigcap_{j=1}^m\St(b_j)$.
By assumption, there is a function $f\colon\Sigma\to\C$ such that, for any sequence $x_n\in B_n$ converging in $\Sigma$, the sequence $f_n(x_n)$ converges to $f(\lim_nx_n)$ in
$\C$. We claim that this statement remains true for the extensions $\widehat{f}_n\colon\Sigma\to\C$. Indeed, let us fix a sequence $x_n$ in $\Sigma$ converging to $x$. For all
$n$, there exist black vertices $b_n^{(1)},\dots,b_n^{(m)}$ such that $x_n$ belongs to the star $\St(b_n^{(j)})$ for $j=1,\dots,m$. Let $b_n$ (resp. $b'_n$) be one of these
vertices where $\Re f_n$ is maximal (resp. minimal) on this set. By definition,
\[
\Re f_n(b'_n)\le\Re\widehat{f}_n(x_n)\le\Re f_n(b_n).
\]
Since $b_n$, $b'_n$ and $x_n$ all belong to the adjacent (or identical) stars $\St(b_n)$ and $\St(b'_n)$ whose diameter is at most $4\delta_n$, and
since $\delta_n$ converges to zero, both sequences $b_n$ and $b'_n$ converge to $x=\lim_n x_n$. By the assumption, $f_n(b_n)$ and $f_n(b'_n)$ both converge to $f(x)$.
By the inequalities displayed above, $\Re\widehat{f}_n(x_n)$ converges to $\Re f(x)$. The same argument shows that $\Im\widehat{f}_n(x_n)$ converges to $\Im f(x)$,
proving the claim.

Lemma~\ref{lemma:conv} asserts that $f\colon\Sigma\to\C$ is continuous and the uniform limit of $\widehat{f}_n\colon\Sigma\to\C$ on any compact.
Since the singular set $S\subset\Sigma$ is discrete and $f$ is continuous, it is now sufficient to check that $f$ is holomorphic on $\Sigma\setminus S$. Hence,
we can restrict ourselves to (simply-connected) domains of a Euclidean atlas for the flat surface $S\subset\Sigma$. In other words,
we can assume that we are working in a simply connected planar domain $U\subset\C$. Finally, by Morera's theorem, it is enough to show that $\int_\gamma f(z)\,dz$ vanishes
for any piecewise smooth loop $\gamma$ in $U$.

So, let $\gamma$ be such a loop, let $\ell$ denote its length, and let $n$ be a fixed index. Each time $\gamma$ meets some star $\St(w_n)$, entering it at a point $x$ and
leaving it at $x'$, replace $\gamma\cap\St(w_n)$ by the path in $\partial\St(w_n)$ realizing the minimal distance in $\partial\St(w_n)$ between $x$ and $x'$. This yields
a new loop $\gamma_n$ contained in $\St_{\Gamma_n}$, the union of all rhombus edges adjacent to black vertices of $\Gamma_n$. By Lemma~\ref{lemma:bound},
its length satisfies
\[
\ell(\gamma_n)=\sum_{w_n\in W_n}\ell(\gamma_n\cap\St(w_n))\le M(\eta)\sum_{w_n\in W_n}|x-x'|\le M(\eta)\,\ell,
\]
where $M(\eta)$ stands for the uniform bound $\frac{2\pi}{\eta\sin(\eta/2)}$. As the diameter of a star $\St(w_n)$ is at most $4\delta_n$, the union of all these
stars meeting $\gamma$ is contained in the tubular neihborhood of $\gamma$ of diameter $8\delta_n$, which also contains the compact set $C_n$ enclosed by $\gamma$ and
$\gamma_n$. Therefore, $\mathit{Area}(C_n)\le 8\delta_n\ell$.

We shall now prove that the sequence $\int_{\gamma_n}f(z)\,dz$ converges to $\int_\gamma f(z)\,dz$. Let us first assume that $f$ is of class $C^1$. In such a case,
$\bar\partial f$ is bounded above by some constant $M$ on the compact set $C$ given by the tubular neighborhood of $\gamma$ of diameter $8\max_n \delta_n$. As $C$ contains all
the $C_n$'s, this yields a uniform bound for $\bar\partial f$ on all $C_n$'s. By Stokes formula and the inequality above,
\begin{eqnarray*}
\left|\int_\gamma f(z)\,dz-\int_{\gamma_n}f(z)\,dz\right|&=&\left|\int_{\partial C_n} f(z)\,dz\right|=\left|\iint_{C_n} \bar\partial f(z)\,dz\wedge d\bar{z}\right|\\
&\le&\iint_{C_n} |\bar\partial f(z)|\,dz\wedge d\bar{z}\le 8\,\delta_n\,\ell\,M,
\end{eqnarray*}
proving the claim in this special case. In the general case of a continuous function $f$, let $g_k$ be a sequence of $C^1$ functions on $U$ converging uniformly to $f$
on every compact. By the inequalities displayed above, we obtain that
\begin{eqnarray*}
\Big|\int_\gamma f-\int_{\gamma_n}f\,\Big|&\le&\Big|\int_\gamma f-\int_\gamma g_k\,\Big|+\Big|\int_\gamma g_k-\int_{\gamma_n}g_k\,\Big|+
\Big|\int_{\gamma_n}g_k-\int_{\gamma_n}f\,\Big|\\
&\le&\ell\,\sup_C|f-g_k|\;+\;8\,\delta_n\,\ell\,M\;+\;M(\eta)\,\ell\,\sup_C|f-g_k|
\end{eqnarray*}
is arbitrarily small, proving the claim.

Recall that $\widehat{f}_n$ converges uniformly to $f$ on the compact $C$. Therefore, for any fixed index $k$,
\[
\left|\int_{\gamma_k}\widehat{f}_n(z)\,dz-\int_{\gamma_k}f(z)\,dz\right|\le M(\eta)\,\ell\,\sup_C|\widehat f_n-f|
\]
is arbitarily small, so $\int_{\gamma_k}\widehat{f}_n$ converges to $\int_{\gamma_k}f$.

We are finally ready to show that $\int_\gamma f(z)\,dz$ is equal to zero. As $\St_{\G_n}$ induces a cellular decomposition of the simply-connected domain $U$, the cycle
$\gamma_n\subset\St_{\G_n}$ is a cellular boundary, that is, $\gamma_n=\partial\left(\sum_{w_n}\St(w_n)\right)$ for some vertices $w_n$. Since $f_n$ is discrete holomorphic,
Lemma~\ref{lemma:Morera} implies
\[
\int_{\gamma_n}\widehat{f}_n(z)\,dz=\sum_{w_n}\int_{\partial\St(w_n)}\widehat{f}_n(z)\,dz=0.
\]
By the three claims above,
\[
\int_\gamma f(z)\,dz=\lim_k\int_{\gamma_k}f(z)\,dz=\lim_k\lim_n\int_{\gamma_k}\widehat{f}_n(z)\,dz=\lim_n\int_{\gamma_n}\widehat{f}_n(z)\,dz=0.
\]
This concludes the proof of the theorem.
\end{proof}

\section{Discrete Dirac operators on Riemann surfaces}
\label{sec:Dirac}

In the previous section, we defined a discrete analog of the $\bar\partial$ operator on functions on a Riemann surface $\Sigma$.
The aim of the present section is to modify this construction, yielding an analog of the Dirac operator $D$ on spinors on $\Sigma$.
Here again, we shall start by giving in Paragraph~\ref{sub:DG'} discretizations of all the geometric objects involved in the definition of $D$ (Table~\ref{table:dic2}).
The actual definition of the discrete Dirac operator is to be found in Paragraph~\ref{sub:Dirac}, while Paragraph~\ref{sub:CLTS} deals with the application
of our convergence theorem to spinors (Theorem~\ref{thm:CLTS}).

\subsection{More discrete geometry}
\label{sub:DG'}

Let us first recall the definition of the Dirac operator on a closed Riemann surface $\Sigma$, referring to~\cite{Ati} for details.
Let $(\varphi_\alpha\colon U_\alpha\to\C)_\alpha$ be an atlas for $\Sigma$, and let
$f_{\alpha\beta}\colon\varphi_\beta(U_\alpha\cap U_\beta)\to\varphi_\alpha(U_\alpha\cap U_\beta)$ denote the corresponding transition functions.
Then, $\kappa_{\alpha\beta}\colon U_\alpha\cap U_\beta\to\C^*$ given by $\kappa_{\alpha\beta}(p)=f'_{\alpha\beta}(\varphi_\beta(p))^{-1}$ is a holomorphic function,
that is, $\kappa_{\alpha\beta}$ is an element of the \v{C}ech cochain group $C^1(\mathcal{U},\mathcal{O}^*)$, where $\mathcal{U}=(U_\alpha)$ and $\mathcal{O}^*$ denotes the sheaf
of non-vanishing holomorphic functions on $\Sigma$. By the chain rule, it is actually a cocycle, so it defines an element in $H^1(\mathcal{U},\mathcal{O}^*)$.
The corresponding holomorphic line bundle $K\in H^1(\Sigma,\mathcal{O}^*)$ is called the {\bf canonical bundle\/} over $\Sigma$. With the notations of
Section~\ref{sec:CA}, $K$ is nothing but the holomorphic cotangent bundle $T^*\Sigma^+$, while $\bar K$ coincides with $T^*\Sigma^-$. Hence, the $\bar\partial$ operator
can be seen as a map $\bar\partial\colon C^\infty(1)\to C^\infty(\bar K)$, where $1$ denotes the trivial line bundle.

The set $\S(\Sigma)$ of {\bf spin structures\/} on $\Sigma$ can be defined as the set of isomorphism classes of holomorphic line bundles that are square roots of $K$, that is,
\[
\S(\Sigma)=\{L\in H^1(\Sigma,\mathcal{O}^*)\,|\,L^2=K\}.
\]
This is easily seen to be an affine space over $H^1(\Sigma;\Z_2)$. Note that a spin structure $L$ is given by a cocycle
$(\lambda_{\alpha\beta})\in Z^1(\mathcal{U},\mathcal{O}^*)$ such that $\lambda_{\alpha\beta}^2=\kappa_{\alpha\beta}$.
Then, a {\bf spinor\/} $\psi\in C^\infty(L)$ can be described by a family of smooth functions $\psi_\alpha\in C^\infty(U_\alpha)$ such that
$\psi_\alpha(p)=\lambda_{\alpha\beta}(p)\,\psi_\beta(p)$ for $p\in U_\alpha\cap U_\beta$. Since $\lambda_{\alpha\beta}$ is holomorphic,
the assignement $(\psi_\alpha)\mapsto(\bar\partial\psi_\alpha)$ defines a map
\[
\bar\partial_L\colon C^\infty(L)\to C^\infty(L\otimes\bar K),
\]
called the {\bf twisted $\bar\partial$ operator\/}. An element $\psi\in C^\infty(L)$ is a {\bf holomorphic spinor\/} if it is in the kernel of $\bar\partial_L$.

Finally, the choice of a hermitian metric on $\Sigma$ allows to define an anti-linear isomorphism $h\colon C^\infty(L\otimes\bar K)\to C^\infty(\bar L)$.
The {\bf Dirac operator\/} is the self-adjoint operator on $C^\infty(L)\oplus C^\infty(\bar L)$ whose restriction to $C^\infty(L)$ is given by
\[
D_L=h\circ\bar\partial_L\colon C^\infty(L)\to C^\infty(\bar L).
\]
By abuse of language, we shall also call $D_L$ the Dirac operator.

\medskip

Let us now give discrete analogs of the objects described above. As explained in the previous section, an analog of a Riemann surface
is a bipartite graph $\Gamma$ embedded in a flat surface $\Sigma$ with cone type singularities supported at $S$, inducing a cell decomposition $X$
of $\Sigma$. Furthermore, in order to define $\bar{\partial}\colon \C^B\to\Omega^1(W)$, we assumed that $\Gamma$ is isoradially embedded
in $\Sigma$, with $S$ contained in $B\cup V(\Gamma^*)$.

By definition, $\Sigma_0:=\Sigma\setminus S$ is endowed with an atlas whose transition functions are Euclidean isometries. Therefore, the associated \v{C}ech cocycle
consists of $S^1$-valued constant functions. This defines an element $K_0$ of $H^1(\Sigma_0;S^1)$.
Using the long exact sequence for the pair $(\Sigma,\Sigma_0)$, one easily checks that $K_0$ is the restriction of a class $[\kappa]\in H^1(\Sigma;S^1)=H^1(X;S^1)$ if and only if
$\exp(i\theta_x)=1$ for all $x\in S$. We shall therefore assume that all cone angles $\theta_x$ are positive multiples of $2\pi$, and call such a cocycle
$\kappa\in Z^1(X;S^1)$ a {\bf discrete canonical bundle\/} over $\Sigma$. Note that the cohomology class of $\kappa$ is uniquely determined by the flat metric on 
$\Sigma$ and the cellular decomposition $X$. Furthermore, such a cocycle $\kappa$ is very easy to compute, as demonstrated by the following remarks.

\begin{remark}
If the flat surface $\Sigma$ has trivial holonomy, one can simply choose $\kappa=1$ as discrete canonical bundle.
\end{remark}

\begin{remark}
\label{rem:P}
It is always possible to represent $\Sigma$ as planar polygons $P$ with boundary identifications. Furthermore, these polygons can be chosen so that
$\G$ intersects $\partial P$ transversally, except at possible singularities in $S\cap B$. Define $\kappa$ by
\[
\kappa(e)=
\begin{cases}
1 & \text{if $e$ is contained in the interior of $P$;} \\
\exp(-i\vartheta) & \text{if $e$ meets $\partial P$ transversally,}
\end{cases}
\]
where $\vartheta$ denotes the angle between the sides of $\partial P\subset\C$ met by the edge $e$. If $S$ is contained in $V(\G^*)$, this
defines completely a natural choice of discrete canonical bundle $\kappa$. If $S\cap B$ is not empty, the partially defined $\kappa$ above can be extended to a cocycle
yielding a discrete canonical bundle.
\end{remark}

\begin{example}
\label{ex:4g}
Let $P$ be the regular $4g$ gon with boundary identification according to the word $\prod_{j=1}^g a_jb_ja_j^{-1}b_j^{-1}$.
This defines a flat metric on the genus $g$ orientable surface $\Sigma_g$ with one singularity of angle $2\pi(2g-1)$.
Given a graph $\G\subset\Sigma_g$ meeting $\partial P$ transversally, the associated canonical bundle is given by
$\kappa(e)=\exp(-i\pi\frac{g-1}{g})$ for edges of $\G$ meeting $\partial P$, and $\kappa(e)=1$ for interior edges.
\end{example}

Finally, note that the assumption that all cone angles are multiples of $2\pi$ simply means that $\Sigma$ has trivial local holonomy.
In such a case, the holonomy defines an element of $\mathrm{Hom}(\pi_1(\Sigma),S^1)=H^1(\Sigma;S^1)=H^1(X;S^1)$, and a representative of this cohomology class
is exactly the inverse of a discrete canonical bundle. Note also that this assumption rules out the 2-sphere from our setting.

Mimicking the continuous case, let us define a {\bf discrete spin structure\/} on $\Sigma$ as any cellular 1-cocycle $\lambda\in Z^1(X;S^1)$ such that $\lambda^2=\kappa$.
(See ~\cite{CR1} for another notion of discrete spin structure, valid for any cellular decomposition of an orientable surface.)
Two discrete spin structures will be called {\bf equivalent\/} if they are cohomologous. The set $\S(X)$ of equivalent classes of discrete spin structures on
$\Sigma$ is then given by
\[
\S(X)=\{[\lambda]\in H^1(X;S^1)\,|\,[\lambda]^2=[\kappa]\}.
\]
One easily checks that this set admits a freely transitive action of the abelian group $H^1(\Sigma;\{\pm 1\})$. In other words, and using additive notations,
$\S(X)$ is an affine $H^1(\Sigma;\Z_2)$-space. Therefore, there exist (non-canonical) $H^1(\Sigma;\Z_2)$-equivariant bijections $\S(X)\to\S(\Sigma)$. Furthermore:

\begin{proposition}
\label{prop:spin}
If all cone angles of $\Si$ are odd multiples of $2\pi$, then there exists a canonical $H^1(\Sigma;\Z_2)$-equivariant bijection $\S(X)\to\S(\Sigma)$.
\end{proposition}
\begin{proof}
Let $\kappa\in Z^1(X;S^1)$ be a fixed discrete canonical bundle over $\Sigma$. For each $\lambda\in Z^1(X;S^1)$ such that $\lambda^2=\kappa$, we shall now construct a vector
field $V_\lambda$ on $\Si$ with zeroes of even index. Such a vector field is well-known to define a spin structure, or equivalently -- by Johnson's theorem~\cite{Joh} --
a quadratic form $q_\lambda$ on $H_1(\Si;\Z_2)$. The proof will be completed with the verification that two equivalent $\lambda$'s induce identical quadratic forms, and that
the assignment $[\lambda]\mapsto q_\lambda$ is $H^1(\Si;\Z_2)$-equivariant.

Let $\lambda\in Z^1(X;S^1)$ be given by $\lambda(e)=\exp(i\beta_\lambda(e))$ with $0\le\beta_\lambda(e)<2\pi$, where $e$ is an edge of $X$ oriented from the white end
to the black end, and set $\beta_\lambda(-e)=-\beta_\lambda(e)$.
\parpic[r]{$\begin{array}{c}
\includegraphics[height=1.6cm]{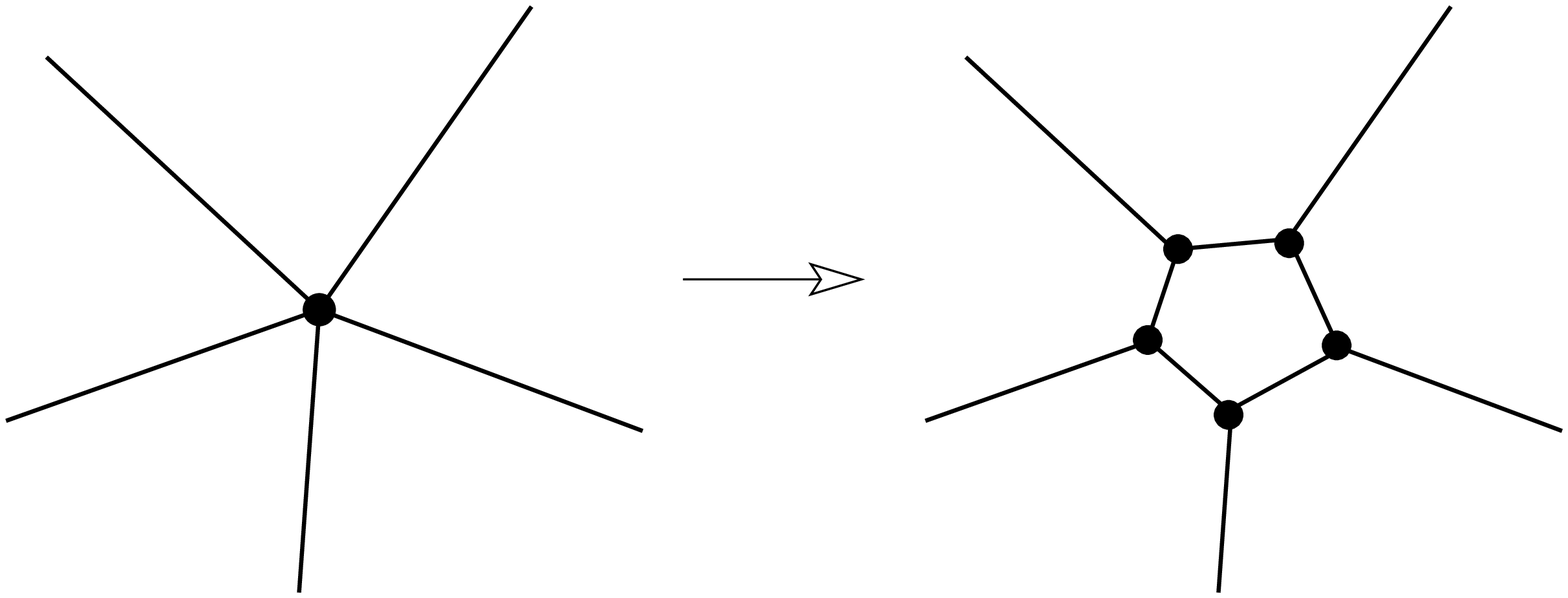}
\end{array}$}
First, replace the cellular decomposition $X$ of $\Si$ by $X'$, where each singularity $b\in B\cap S$ is
removed as illustrated opposite. Obviously, $\lambda$ induces $\lambda'\in Z^1(X';S^1)$ by setting $\lambda'(e)=1$ for each newly created edge $e$. Now, fix an arbitrary
tangent vector $V_\lambda(w)$ at some white vertex $w$, and extend it to the 1-skeleton $\Gamma'$ of $X'$ as follows: running along an edge $e$ oriented from the white end
to the black end, rotate the tangent vector by an angle of $2\beta_\lambda(e)$ in the negative direction. (On the newly created edges, just extend the vector field without
any rotation.) As $\lambda'$ is a cocycle and each cone angle is a multiple of $2\pi$, this gives a well-defined vector field along $\G'$. Extend it to the whole surface $\Si$
by the cone construction. The resulting vector field $V_\lambda$ has one zero in the center of each face of $X$, and at each $b\in B\cap S$. One easily checks that such a zero
is of even index if and only if the corresponding cone angle is an odd multiple of $2\pi$, which we assumed.

Following \cite{Joh}, the quadratic form $q_\lambda\colon H_1(\Si;\Z_2)\to\Z_2$ corresponding to $V_\lambda$ is determined as follows: for any regular oriented
simple closed curve $C\subset\Sigma\setminus S$, $q_\lambda([C])+1$ is equal to the winding number of the tangential vector field along $C$ with
respect to the vector field $V_\lambda$. For an oriented simple closed curve $C\subset\G$, we obtain the following equality modulo 2:
\begin{eqnarray*}
q_\lambda([C])&=&1+\frac{1}{2\pi}\Big(\sum_{e\subset C} 2\beta_\lambda(e)+\sum_{v\in C}(\pi-\alpha_v(C))\Big)\\
&=&1+\frac{|C|}{2}+\frac{1}{2\pi}\Big(\sum_{e\subset C} 2\beta_\lambda(e)-\sum_{v\in C}\alpha_v(C)\Big),
\end{eqnarray*}
\parpic[r]{$\begin{array}{c}
\labellist\small\hair 2.5pt
\pinlabel {$v$} at 168 18
\pinlabel {$\alpha_v(C)$} at 175 160
\pinlabel {$C$} at 420 190
\endlabellist
\includegraphics[height=1.2cm]{int}
\end{array}$}
\noindent where the first sum is over all {\em oriented\/} edges in the oriented curve $C$, and $\alpha_v(C)$ is the angle illustrated opposite. Obviously, equivalent $\lambda$'s
induce the same quadratic form $q_\lambda$. Finally, given two discrete spin structures $\lambda_1,\lambda_2$, the cohomology class of the 1-cocycle
$\lambda_1/\lambda_2\in Z^1(X;\{\pm 1\})$ is determined by its value on oriented simple closed curves in $\G$. For such a curve $C$, we have
\[
(\lambda_1/\lambda_2)(C)=\exp\Big(i\sum_{e\subset C}(\beta_{\lambda_1}(e)-\beta_{\lambda_2}(e))\Big)=\exp\big(i\pi(q_{\lambda_1}-q_{\lambda_2})([C])\big).
\]
Therefore, the assignement $[\lambda]\mapsto q_\lambda$ is $H^1(\Si;\Z_2)$-equivariant, which concludes the proof.
\end{proof}

\begin{remark}
If the flat surface $\Sigma$ has trivial holonomy, then $[\kappa]$ is trivial, so the set $\S(X)$ is equal to the $2g$-dimensional vector space $H^1(\Sigma;\Z_2)$.
\end{remark}

\begin{remark}
Let $\G\subset\Sigma$ be described via planar polygons as explained in Remark~\ref{rem:P}, and let us assume that the singular set $S$ is contained in $V(\G^*)$.
In such a case, a discrete spin structure is given by
\[
\lambda(e)=
\begin{cases}
1 & \text{if $e$ is contained in the interior of $P$;} \\
\exp(-i\vartheta/2) & \text{if $e$ meets $\partial P$,}
\end{cases}
\]
where $\exp(-i\vartheta/2)$ denotes one of the square roots of the angle between the sides of $\partial P\subset\C$ met by the edge $e$.

\end{remark}

\begin{example}
For $\Sigma$ as in Example~\ref{ex:4g} above, equivalence classes of spin structures correspond to the $2^{2g}$ choices of $2g$ square roots of 
$\exp(-i\pi\frac{g-1}{g})$, one for each pair of boundary edges of $P$.
In particular, for the flat torus, $\S(X)$ corresponds to the 4 possible choices of 2 square roots of the unity.
\end{example}

Let us now turn to spinors. Given a spin structure $L\in\S(\Si)$, the universal covering $\pi\colon\ws\to\Si$ induces the following pullback diagram:
\[
\xymatrix{
E \ar[r]^{\Pi} \ar[d]_{\pi^*p} & L\ar[d]^p  \\
\ws \ar[r]^\pi & \Sigma.
}
\]
By the lifting property of the covering map $\Pi\colon E\to L$, any spinor $\psi\in C^\infty(L)$ induces a section $\widetilde{\psi}\in C^\infty(E)$
such that $\Pi\circ\widetilde{\psi}=\psi\circ\pi$, unique up to the action of $\pi_1(\Si)$. Since $\ws$ is contractible (recall that the 2-sphere is ruled out by the assumption
on the cone angles), the line bundle $E\to\ws$ is trivial. Hence,
$\widetilde{\psi}$ is really a complex-valued function on $\ws$ satisfying some $\pi_1(\Si)$-periodicity property depending on $L$.
This alternative point of view on spinors leads to the following definition.

Let $\lambda\in Z^1(X;S^1)$ be a discrete spin structure on $\Si$, and let $\pi\colon\widetilde X\to X$ denote the cellular map given by the universal covering of $\Si$.
Note that the bipartite structure on $\G$ lifts to a bipartite structure $V(\widetilde\G)=\widetilde B\cup \widetilde W$ on $\widetilde\G=\pi^{-1}(\G)$.
The space $C(\lambda)$ of {\bf discrete spinors} is the set of all $\psi\in\C^{\widetilde B}$ such that, for any $b,b'\in\widetilde B$ with $\pi(b)=\pi(b')$,
\[
\psi(b')=\lambda(\pi(\gamma_{b,b'}))\,\psi(b),
\]
where $\gamma_{b,b'}$ denotes a path in $\widetilde\G$ from $b$ to $b'$. As $\lambda$ is a cocycle and $\widetilde\Si$ is simply-connected, this condition does not depend
on the choice of such a path. Furthermore, equivalent discrete spin structures $\lambda\sim\lambda'$ will yield the same space $C(\lambda)=C(\lambda')$.
Note that the choice of any fundamental domain $P\subset\ws$ for the action of $\pi_1(\Si)$ yields an identification $C(\lambda)\stackrel{\varphi_P}{\simeq}\C^B$.
However, this identification is not canonical, unless $\lambda$ is trivial.

Similarly, let us define the space $C(\bar\lambda)$ as the set of all $\psi\in\C^{\widetilde W}$ such that $\psi(w')=\lambda(\pi(\gamma_{w,w'}))\,\psi(w)$ whenever
$\pi(w)=\pi(w')$. Here again, a fundamental domain $P\subset\ws$ yields a non-canonical identification $C(\bar\lambda)\stackrel{\bar\varphi_P}{\simeq}\C^W$.

Finally, and for reasons that will become clear in the next paragraph, the role of the hermitian metric on $\Sigma$ will be played by a nowhere vanishing vector field
$V\in\X(W)$ along the white vertices of $\G$. Furthermore, we shall normalize this vector field so that $|V_w|=2\mathit{Area}(\St(w))$ for all $w\in W$.

Table~\ref{table:dic2} summarizes the second part of our dictionary.

\begin{table}[t]\centering
\caption{Discretization dictionary, part 2}
\label{table:dic2}

\setlength\extrarowheight{4pt}
\setlength\arrayrulewidth{.8pt}

\begin{tabulary}{\textwidth}{||C||C||}\hhline{|t:=:t:=:t|}

\em{the geometric object} & \em{the discrete analog} \\ \hhline{|:=::=:|}

the canonical bundle $K\in H^1(\Sigma,\mathcal{O}^*)$ &
the cohomology class $[\kappa]\in H^1(X;S^1)$, provided all cone angles $\{\theta_x\}_{x\in S}$ are multiples of $2\pi$ \\ \hhline{||-||-||}

the affine $H^1(\Sigma;\Z_2)$-space $\mathcal{S}(\Sigma)$ of spin structures &
the affine $H^1(\Sigma;\Z_2)$-space $\mathcal{S}(X)$ of equivalence classes of square roots of $\kappa\in Z^1(X;S^1)$ \\ \hhline{||-||-||}

the spinors $C^\infty(L)$ associated to $L\in\mathcal{S}(\Sigma)$  &
the space $C(\lambda)\subset\C^{\widetilde B}$ of discrete spinors associated to $\lambda\in\S(X)$ \\ \hhline{||-||-||}

$C^\infty(\bar L)$ & $C(\bar\lambda)\subset\C^{\widetilde W}$ \\ \hhline{||-||-||}

a hermitian metric on $\Sigma$ & a (normalized) vector field $V\in\X(W)$ \\ \hhline{||-||-||}

the twisted $\bar\partial$ operator $\bar\partial_L\colon C^\infty(L)\to C^\infty(L\otimes\bar K)$ & 
$\bar\partial_\lambda\colon C(\lambda)\to\Omega^1(\widetilde W)$ given by the restriction of $\bar\partial$ to $C(\lambda)\subset\C^{\widetilde B}$ \\ \hhline{||-||-||}

the Dirac operator $D_L\colon C^\infty(L)\to C^\infty(\bar L)$ & 
$D_\lambda\colon \C^B\simeq C(\lambda)\to C(\bar\lambda)\simeq\C^W$ given by $\bar\partial_\lambda$ evaluated along the vector field $V$ \\ \hhline{|b:=:b:=:b|}

\end{tabulary}
\end{table}

\subsection{The discrete Dirac operators}
\label{sub:Dirac}

As above, let $\G$ be a bipartite graph isoradially embedded in a flat surface $\Sigma$ with cone type singularities $S\subset B\cup V(\G^*)$, and let us assume
that all cone angles are multiples of $2\pi$. Note that all these structures lift to the universal cover $\pi\colon\ws\to\Si$. Indeed, this map defines a bipartite graph
$\widetilde \G$ isoradially embedded in the flat surface $\ws$ with cone type singularities $\widetilde S\subset\widetilde B\cup V(\widetilde\G^*)$.
Let us define the {\bf discrete twisted $\bar\partial$ operator\/} associated to $\lambda\in\S(X)$ is the $\C$-linear map
\[
\bar\partial_\lambda\colon C(\lambda)\to\Omega^1(\widetilde W)=\prod_{w\in \widetilde W}\X(w)^*
\]
defined by the restriction of the discrete $\bar\partial$ operator $\bar\partial\colon\C^{\widetilde B}\to\Omega^1(\widetilde W)$ to $C(\lambda)\subset\C^{\widetilde B}$.

We need a map $\Omega^1(\widetilde W)\to\C^{\widetilde W}$ discretizing the anti-linear isomorphism $C^\infty(L\otimes\bar K)\to C^\infty(\bar L)$
induced by a hermitian metric. A discrete hermitian metric, that is, a normalized vector field $V\in\X(W)$ induces a very natural such map, namely the evaluation at
$\widetilde V\in\X(\widetilde W)$, the lift of $V$ to $\widetilde W$. Putting all the pieces together yields the map $D'_\lambda\colon C(\lambda)\to\C^{\widetilde W}$ given by
\[
(D'_\lambda\psi)(\tilde w)=\sum_{\tilde b\sim \tilde w}\nu(\tilde w,\tilde b)e^{i\vartheta_{\widetilde V}(\tilde w,\tilde b)}\,\psi(\tilde b),
\]
with the notations of Section~\ref{sub:DDol}. One easily checks that the image of $D'_\lambda$ is contained in $C(\bar\lambda)$, and that equivalent discrete spin
structures $\lambda\sim\lambda'$ induce identical maps $D'_\lambda=D'_{\lambda'}\colon C(\lambda)\to C(\bar\lambda)$. Finally, the following lemma provides us with a less
cumbersome definition of this operator.

\begin{lemma}
Pick a simply-connected fundamental domain $P\subset\ws$ for the action of $\pi_1(\Si)$, and let $D_{\lambda}\colon\C^B\to\C^W$ be the composition
$\bar\varphi_P\circ D'_\lambda\circ\varphi_P^{-1}$. Then, for a well-chosen representative of $[\lambda]\in\S(X)$,
\[
(D_{\lambda}\psi)(w)=\sum_{b\sim w}\lambda(w,b)\nu(w,b)e^{i\vartheta_V(w,b)}\,\psi(b)
\]
for $\psi\in\C^B$ and $w\in W$.
\end{lemma}
\begin{proof}
Fix $\psi\in\C^B$, $w\in W$, and let $\tilde w$ denote the element of $\pi^{-1}(w)$ in $P$. Then,
\begin{eqnarray*}
(D_{\lambda}\psi)(w)&=& D'_\lambda(\varphi_P^{-1}(\psi))(\tilde w)\cr
&=&\sum_{\tilde b\sim \tilde w}\nu(\tilde w,\tilde b)e^{i\vartheta_{\widetilde V}(\tilde w,\tilde b)}\,\varphi_P^{-1}(\psi)(\tilde b)\cr
&=&\sum_{\tilde b\sim \tilde w}\nu(\tilde w,\tilde b)e^{i\vartheta_{\widetilde V}(\tilde w,\tilde b)}\,\lambda(\pi(\gamma_{\tilde b',\tilde b}))\psi(\pi(\tilde b)),
\end{eqnarray*}
where $\tilde b'$ denotes the element of $P$ such that $\pi(\tilde b')=\pi(\tilde b)$. As $P$ is simply-connected, there exists a representative $\lambda$
such that $\lambda(e)=1$ for any edge $e$ contained in the interior of $\pi(P)$. Setting $\pi(\tilde b)=b$, we get
\[
(D_{\lambda}f)(w)=\sum_{b\sim w}\nu(w,b)e^{i\vartheta_V(w,b)}\,\lambda(w,b)\,\psi(b),
\]
what was to be shown.
\end{proof}

This discussion motivates the following definition.

\begin{definition}\label{def:Dirac}
Let $\Gamma$ be a bipartite graph isoradially embedded in a flat surface $\Sigma$ with conical singularities $S\subset V(\Gamma^*)\cup B$, and all cone angles
multiples of $2\pi$. Given any discrete spin structure $\lambda$, the associated {\bf discrete Dirac operator\/} is the map
$D_\lambda\colon\C^B\to\C^W$ defined by
\[
(D_\lambda\psi)(w)=\sum_{b\sim w}\lambda(w,b)\nu(w,b)e^{i\vartheta_V(w,b)}\,\psi(b)
\]
for $\psi\in\C^B$ and $w\in W$.
The sum is over all vertices $b$ adjacent to $w$, $\nu(w,b)$ denote the length of the edge dual to $(w,b)$, and $\vartheta_V(w,b)$
is the angle at $w\in W$ illustrated in Figure~\ref{fig:theta}.

A discrete spinor $\psi\in\C^B$ is {\bf discrete holomorphic (with respect to $\lambda$)\/} if $D_\lambda\psi=0$.
\end{definition}

Note that $D_\lambda$ is essentially independant from the choice of the discrete hermitian metric $V$: another choice would yield the matrix $QD_\lambda$,
where $Q$ is a diagonal matrix with diagonal coefficients in $S^1$. Furthermore, if $\lambda$ and $\lambda'$ are equivalent discrete spin structures, then
there exist two such matrices $Q,Q'$ such that $D_{\lambda'}=QD_\lambda Q'$.

\begin{remark}
The map $D_\lambda\colon\C^B\to\C^W$ defined above is really the discrete analog of the restriction of the Dirac operator to $C^\infty(L)$. The full Dirac
operator on $C^\infty(L)\oplus C^\infty(\bar L)$ being self-adjoint, it would discretize to the operator on $\C^{V(\G)}=\C^B\oplus\C^W$ given by the matrix
$\begin{pmatrix}0 & D_\lambda^*\cr D_\lambda &0  \end{pmatrix}$.
\end{remark}

\begin{remark}
\label{rem:W}
We have assumed throughout the paper that no white vertex of $\G$ is a conical singularity of $\Si$. This was crucial in Section~\ref{sec:CA} in order to define the
discrete $\bar\partial$ operator. However, in the present section, we could have dropped this condition and defined $D_\lambda$ using any choice of a direction at each
$w\in W$ (for example, given by a perfect matching). All the results of the paper, apart from the ones of Section~\ref{sec:CA}, still hold in this slightly more general setting.
\end{remark}

\subsection{The convergence theorem for spinors}
\label{sub:CLTS}

Let us conclude this section with the application of the convergence theorem (Theorem~\ref{thm:CL}) to spinors.

Let $\G_n$ be a sequence of graphs embedded in a flat surface $\Si$, and let $\lambda_n\in\S(X_n)$ be discrete spin structures inducing the same spin structure $L\in\S(\Si)$.
(Recall that by Proposition~\ref{prop:spin}, there is a canonical equivariant bijection $\S(X)\to\S(\Si)$ provided all cone angles are odd multiples of $2\pi$.)
We shall say that a sequence $\psi_n\in C(\lambda_n)\subset\C^{\widetilde B_n}$ of discrete spinors {\bf converges} to a section $\psi$ of the line bundle $L\to\Si$ if,
for some lift $\widetilde\psi\colon\ws\to\C$ of $\psi$, the following holds: for any sequence $\tilde x_n\in\widetilde B_n$ converging to $\tilde x\in\ws$,
$\psi_n(\tilde x_n)$ converges to $\widetilde\psi(\tilde x)$.

\begin{theorem}
\label{thm:CLTS}
Let $\Sigma$ be a flat surface with conical singularities supported at $S$ whose angles are odd multiples of $2\pi$.
Consider a sequence $\Gamma_n$ of bipartite graphs isoradially embedded in $\Sigma$ with $S\subset V(\Gamma_n^*)\cup B_n$,
inducing cellular decompositions $X_n$ of $\Sigma$. Assume that the radii $\delta_n$ of $\G_n$ converge to $0$, and that there is some $\eta>0$ such that
all rhombi angles of all these $\Gamma_n$'s belong to $[\eta,\pi-\eta]$. Finally, pick a sequence of discrete spin structures
$\lambda_n\in Z^1(X_n;S^1)$ inducing the same class in $H^1(\Sigma;S^1)$, and let $L\in\S(\Sigma)$
denote the corresponding spin structure on $\Sigma$.

Let $\psi_n\in C(\lambda_n)$ be a sequence of discrete spinors converging to a section $\psi$ of the line bundle $L\to\Si$.
If for each $n$, $\psi_n$ is discrete holomorphic with respect to $\lambda_n$, then $\psi$ is a holomorphic spinor.
\end{theorem}

\begin{proof}
By assumption, $\psi_n\in\C^{\widetilde B_n}$ are discrete holomorphic functions on $\ws$ converging to $\widetilde\psi\colon\ws\to\C$ in the sense of Theorem~\ref{thm:CL}.
By this result, $\widetilde\psi$ is a holomorphic function. Therefore, $\psi\in C^\infty(L)$ is a holomorphic spinor.
\end{proof}

\section{Relation to Kasteleyn matrices and the dimer model}
\label{sec:Kast}

Recall that a {\bf dimer covering\/}, or {\bf perfect matching\/} on a finite connected graph $\Gamma$ is a family $M$ of edges 
of $\Gamma$, called {\bf dimers\/}, such that each vertex of $\Gamma$ is adjacent to exactly one dimer.
Any edge weight system $\nu\colon E(\G)\to[0,\infty)$ induces a probability measure $\mu$ on the set $\M(\G)$
of dimer coverings of $\Gamma$. It is given by
\[
\mu(M)=\frac{\nu(M)}{Z(\G,\nu)},
\]
where $\nu(M)=\prod_{e\in M}\nu(e)$ and
\[
Z(\G,\nu)=\sum_{M\in\M(\G)}\nu(M)
\]
is the associated {\bf partition function\/}. The study of this measure is called the {\bf dimer model\/} on $\G$.

The aim of this section is to relate the discrete Dirac operators introduced above to some matrices, called Kasteleyn matrices, which provide a standard tool for the
dimer model on a graph.

\subsection{Kasteleyn flatness}
\label{sub:flat}

Let $\G$ be a finite bipartite graph. Fix a field $\mathbb{F}$ containing $\R$ as a subfield,
and let $G$ be a multiplicative subgroup of $\mathbb{F}^*$ containing $\{\pm 1\}$. (The examples to keep in mind
are $G=\{\pm 1\}\subset\R^*$ and $G=S^1\subset\C^*$.)
Since each edge of $\G$ is endowed with a natural orientation (say, from the white vertex
to the black one), a map $\omega\colon E(\G)\to G$ can be viewed as a cellular 1-cochain $\omega\in C^1(\G;G)$, where
$\omega\left(\begin{array}{c}\includegraphics[width=0.5in]{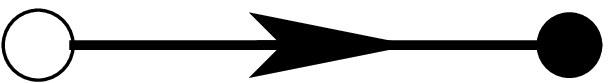}\end{array}\right)=\omega(e)$ and
$\omega\left(\begin{array}{c}\includegraphics[width=0.5in]{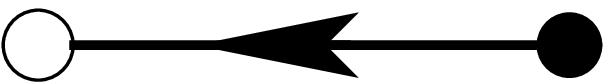}\end{array}\right)=\omega(e)^{-1}$.

Let us order the set $B$ of black vertices of $\G$, as well as the white vertices $W$, and fix a cochain 
$\omega\in C^1(\G;G)$. Let $K^\omega=K^\omega(\G,\nu)$ denote the associated weighted bipartite
adjacency matrix: This is the $(|W|\times |B|)$-matrix with coefficients in $\mathbb{F}$ defined by
\[
(K^\omega)_{w,b}=\sum_e\nu(e)\omega(e),
\]
the sum being on all edges $e$ of $\G$ joining $w\in W$ and $b\in B$.

The goal is now to find cochains $\omega$ so that $\det(K^\omega(\G,\nu))$ can be used to compute $Z(\G,\nu)$.
Embed $\G$ in an oriented closed surface $\Sigma$ so that $\Sigma\setminus\G$ consists of open 2-discs (this is always possible), and let $X$ denote the 
induced cellular decomposition of $\Sigma$.  The {\bf Kasteleyn curvature\/} of $\omega\in C^1(\G;G)$ at a face $f$ of $X$ is the element of $G$ defined by
\[
c_\omega(f):=(-1)^{\frac{\left|\partial f\right|}{2}+1}\omega(\partial f),
\]
where $\partial f$ denotes the oriented boundary of the oriented face $f$, and $\left|\partial f\right|$ the number
of edges in $\partial f$. This defines a curvature 2-cochain $c_\omega\in C^2(X;G)$.
A 1-cochain $\omega$ is said to be {\bf Kasteleyn flat\/} (or simply {\bf flat\/}) if $c_\omega$ is equal to $1$.
Finally, we shall say that two cochains $\omega,\omega'\in C^1(\G;G)$ are {\bf gauge equivalent\/} (or simply
{\bf equivalent\/}) if they are cohomologous, that is, if they can be
related by iterations of the following transformation: pick a vertex of $\G$ and multiply all adjacent edge weights
by some $g\in G$. Note that equivalent cochains $\omega,\omega'$ have the same curvature, and that the determinant
of the associated matrices $K^\omega$ and $K^{\omega'}$ differ by multiplication by an element of $G$.

\begin{example}\label{ex:Kast}
If $G$ is the multiplicative group $\{\pm 1\}$, then elements of $C^1(\G;G)$ are nothing but orientations of the edges
of $\G$: an edge $e$ is oriented from the white vertex to the black one if and only if $\omega(e)=+1$. Furthermore,
$\omega$ is flat if and only if the corresponding orientation satisfies the following condition: for each face $f$,
the number of boundary edges oriented from black to white has the parity of $\frac{\left|\partial f\right|}{2}+1$.
This is usually called a {\bf Kasteleyn orientation\/}, and the associated matrix $K^\omega$ is called a
{\bf Kasteleyn matrix\/}. By abuse of language, we shall say that two Kasteleyn matrices are equivalent if the corresponding
Kasteleyn orientations are.
\end{example}

\begin{proposition}\label{prop:torsor}
There exists a flat $G$-valued 1-cochain on a bipartite graph $\Gamma\subset\Sigma$ if and only if
$\Gamma$ has an even number of vertices. In this case, the set of equivalence classes of such 1-cochains
is an $H^1(\Sigma;G)$-torsor, that is: it admits a freely transitive action of the abelian group $H^1(\Sigma;G)$.
\end{proposition}
\begin{proof}
Let $V$ (resp. $E$, $F$) denote the number of vertices (resp. edges, faces) of $X$. Given any $\omega\in C^1(\G;G)$,
we have
\[
\prod_{f\subset X}c_\omega(f)=(-1)^{\sum_{f\subset X}\left(\frac{\left|\partial f\right|}{2}+1\right)}=(-1)^{E+F}=(-1)^V,
\]
since the Euler characteristic $\chi(\Sigma)=V-E+F$ is even. Therefore, if $\omega$ is flat, then $V$ is even.
Conversely, if $V$ is even, then $\prod_{f\subset X}c_\omega(f)=1$. This implies that
$c_\omega$ is a coboundary, that is, there exists a $\phi\in C^1(X;G)$ such that $c_\omega=\delta\phi^{-1}$.
Consider now the 1-cochain $\phi\omega$ defined by $(\phi\omega)(e)=\phi(e)\omega(e)$.
Given any face $f$ of $X$, we have the following equality in $G$:
\[
(\delta\phi)(f)=\phi(\partial f)=c_{\phi\omega}(f)c_\omega(f)^{-1}.
\]
Since $c_\omega=\delta\phi^{-1}$, it follows that $c_{\phi\omega}=1$, that is, $\phi\omega$ is flat.

Let us now prove the second statement, assuming that there exists a flat cochain.
Define the action of an element $[\phi]\in H^1(\Sigma;G)=H^1(X;G)$ on $[\omega]$ by $[\phi]\cdot[\omega]=[\phi\omega]$.
Since $\phi$ is a cocycle, the equation displayed above shows that $\phi\omega$ is flat if and only if
$\omega$ is. Note also that $\phi\omega$ is gauge equivalent to $\omega$ if and only if $\phi$ is a coboundary.
Therefore, this action of $H^1(\Sigma;G)$ on the set of equivalence classes is well-defined, and free. Finally, given two
flat systems $\omega$ and $\omega'$, let $\phi$ denote the 1-cochain defined by $\phi(e)=\omega'(e)\omega(e)^{-1}$. Obviously, $\omega'=\phi\omega$, and $\phi$ is a cocycle by the identity displayed above. Therefore, the action
is freely transitive.
\end{proof}

\subsection{Computing the dimer partition function}
\label{sub:partition}

The point of introducing flat cochains is that they can be used to compute the partition function $Z(\Gamma,\nu)$ of
the dimer model, as follows. Note that $Z(\Gamma,\nu)$ is zero unless $\G$ has the same number of white and black vertices, which we shall assume throughout this section.

Let $\mathcal{B}=\{\alpha_j\}$ be a set of simple closed curves on $\Sigma$, transverse to $\Gamma$, whose classes form a
basis of $H_1(\Sigma;\Z)$. For each $\alpha_j\in\mathcal{B}$, let $C_j$ denote the oriented 1-cycle in $\Gamma$
having $\alpha_j$ to its immediate left, and meeting every vertex of $\G$ adjacent to $\alpha_j$ on this side.
Let $\tau$ denote the flat cochain (unique up to equivalence, by Proposition~\ref{prop:torsor}) such that $\tau(C_j)=(-1)^{|C_j|/2+1}$
for all $j$. Let $\omega\in C^1(\Gamma;G)$ be any flat cochain, and let $\varphi$ be the unique element in
$H^1(\Sigma;G)$ such that $\varphi\cdot[\tau]=[\omega]$. Finally, for any
$\eps=(\eps_1,\dots,\eps_{2g})\in\Z_2^{2g}$, let $\omega_\eps$ denote the flat cochain obtained from $\omega$
as follows: multiply $\omega(e)$ by $-1$ each time the edge $e$ meets $\alpha_j$ with $\eps_j=1$.

\begin{theorem}\label{thm:Pf-K}
For any $\alpha\in H_1(\Sigma;\Z)$, let $Z^\mathcal{B}_\alpha(\Gamma,\nu)$ denote the partial partition
function defined by
\[
Z^\mathcal{B}_\alpha(\Gamma,\nu)\;=
\sum_{\genfrac{}{}{0pt}{}{M\in\M(\G)}{\alpha_i\cdot M=\alpha_i\cdot\alpha\;\forall i}}\nu(M).
\]
Then, the following equality holds in $\mathbb{F}$ up to multiplication by an element of $G$:
\[
\sum_{\alpha\in H_1(\Sigma;\Z)}\varphi(\alpha)Z^\mathcal{B}_\alpha(\Gamma,\nu)=
\frac{1}{2^g}\sum_{\eps\in\Z_2^{2g}}(-1)^{\sum_{i<j}\eps_i\eps_j\alpha_i\cdot\alpha_j}\det(K^{\omega_\eps}).
\]
\end{theorem}

Taking the flat cochain $\omega=\tau$ immediately yields:

\begin{corollary}\label{cor:Pf}
The partition function of the dimer model on $\G$ is given by
\[
Z(\Gamma,\nu)\;\dot{=}\;
\frac{1}{2^g}\sum_{\eps\in\Z_2^{2g}}(-1)^{\sum_{i<j}\eps_i\eps_j\alpha_i\cdot\alpha_j}\det(K^{\tau_\eps}).
\]
where $\dot{=}$ stands for equality in $\mathbb{F}$ up to multiplication by an element of $G$.\qed
\end {corollary}

\begin{example}
Assume that the bipartite graph $\G$ is planar. In such a case, one can take $\Sigma$ to be the 2-sphere, so all
flat $G$-valued cochains are equivalent by Proposition~\ref{prop:torsor}. Corollary~\ref{cor:Pf} gives the equality
\[
Z(\G,\nu)\;\dot{=}\;\det(K^\omega)
\]
for any flat $\omega\in C^1(\G;G)$. The case $G=\{\pm 1\}$ is the celebrated Kasteleyn Theorem \cite{Ka1,Ka2,Ka3}.
The mild generalization stated in this example is not truely original, as it easily follows from the discussion in
Section II of \cite{Kup}.
\end{example}

The general formula stated in Theorem~\ref{thm:Pf-K} can seem somewhat cumbersome. Therefore,
let us illustrate its usefulness before giving the proof.

\begin{example}
Let $\mathbb{F}$ be the quotient field of the group ring $\Z[H_1(\Sigma;\Z)]$, and let $G$ denote the subgroup of
$\mathbb{F}^*$ given by $G=\pm H_1(\Sigma;\Z)$. If one chooses a family of curves $\mathcal{B}=\{\alpha_i\}$
as above and denotes by $a_i\in H_1(\Sigma;\Z)$ the class of $\alpha_i$, then $\mathbb{F}=\mathbb{Q}(a_1,\dots,a_{2g})$.
Let $\tau\in C^1(\G;\{\pm 1\})$ be as described above, and consider the cochain $\omega\in C^1(\G;G)$ given by
$\omega(e)=\tau(e)\prod_i a_i^{\alpha_i\cdot e}$. In other words, an edge is multiplied by $a_i$ (resp. $a_i^{-1}$) each
time it crosses $\alpha_i$ in the positive (resp. negative) direction. Then, Theorem~\ref{thm:Pf-K} yields the following
equality in $\mathbb{Q}(a_1,\dots,a_{2g})$, up to multiplication by $\pm a_1^{m_1}\cdots a_{2g}^{m_{2g}}$:
\[
\sum_{n\in\Z^{2g}}Z^\mathcal{B}_n(\Gamma,\nu)\,a_1^{n_1}\cdots a_{2g}^{n_{2g}}\,\;\dot{=}\;\,
\frac{1}{2^g}\sum_{\eps\in\Z_2^{2g}}(-1)^{\sum_{i<j}\eps_i\eps_j\alpha_i\cdot\alpha_j}\det(K^{\omega_\eps}),
\]
where $Z^\mathcal{B}_n(\Gamma,\nu)$ is the partial partition function given by the contribution of all $M\in\M(\G)$
such that $\alpha_i\cdot M=n_i$ for all $i$.
\end{example}

\begin{proof}[Proof of Theorem~\ref{thm:Pf-K}]
Consider a $\pm 1$-valued cochain $\tau$, and interpret it as an orientation of the edges of $\G$ as explained in
Example~\ref{ex:Kast}. One easily checks that the equation $\tau(C_j)=(-1)^{\frac{|C_j|}{2}+1}$ is equivalent
to the following fact: the number of edges in $C_j$ where $\tau$ disagrees with a given orientation on $C_j$ is odd.
By \cite[Theorem~3.9]{Cim}, we have the following equality in $\mathbb{F}$
\[
Z(\Gamma,\mathrm{w})=\frac{\pm 1}{2^g}
\sum_{\eps\in\Z_2^{2g}}(-1)^{\sum_{i<j}\eps_i\eps_j\alpha_i\cdot\alpha_j}\det(K^{\tau_\eps}(\G,\mathrm{w})),
\]
for any $\mathbb{F}$-valued edge weight system $\mathrm{w}$.
Given any cochains $\sigma,\phi\in C^1(\G;G)$, the equality
\[
K^\sigma(\G,\phi\nu)=K^{\phi\sigma}(\G,\nu)
\]
is obvious. Furthermore, if $\phi$ is a cocycle, we shall check shortly that
\[
Z(\G,\phi\nu)=\sum_{M\in\M(\G)}\phi(M)\nu(M)\;\dot{=}\;
\sum_{\alpha\in H_1(\Sigma;\Z)}[\phi](\alpha)Z^\mathcal{B}_\alpha(\Gamma,\nu),
\]
where $[\phi]\in H^1(\Sigma;G)=\mathrm{Hom}(H_1(\Sigma;\Z),G)$ is the cohomology class of $\phi$, and
$Z^\mathcal{B}_\alpha(\Gamma,\nu)$ is the partial partition function defined in the statement of the theorem.
Applying the three equalities displayed above to the weight system $\mathrm{w}=\phi\nu$, with $\phi$ such that
$\phi\tau=\omega$, yields the theorem.

It remains to check the last equation displayed above. Let $\{\alpha_j^*\}$ be the basis in $H_1(\Sigma;\Z)$ dual
to $\mathcal{B}=\{\alpha_i\}$ with respect to the intersection pairing, that is, such that
$\alpha_i\cdot\alpha_j^*=\delta_{ij}$. The difference of any two dimer coverings $M,M_0$ viewed as elements of
$C_1(\G;\Z)$ is clearly a cycle. Since the expression of an arbitrary $\alpha\in H_1(\Sigma;\Z)$ in the basis
$\{\alpha_i^*\}$ is given by $\alpha=\sum_i(\alpha_i\cdot\alpha)\alpha_i^*$, we get
\[
\frac{\phi(M)}{\phi(M_0)}=\phi(M-M_0)=[\phi]\Big(\sum_i\big(\alpha_i\cdot(M-M_0)\big)\alpha_i^*\Big)=
\frac{[\phi]\left(\sum_i(\alpha_i\cdot M)\alpha_i^*\right)}{[\phi]\left(\sum_i(\alpha_i\cdot M_0)\alpha_i^*\right)}.
\]
This implies the equality
\begin{eqnarray*}
\sum_{M\in\M(\G)}\phi(M)\nu(M)&\;\dot{=}\;&\sum_{M\in\M(\G)}[\phi]\Big(\sum_i(\alpha_i\cdot M)\alpha_i^*\Big)\,\nu(M)\\
&=&\sum_{\alpha\in H_1(\Sigma;\Z)}[\phi](\alpha)Z^\mathcal{B}_\alpha(\Gamma,\nu),
\end{eqnarray*}
which concludes the proof.
\end{proof} 

\subsection{Discrete Dirac operators and Kasteleyn matrices}
\label{sub:DK}

Now, let us turn back to our discrete Dirac operators. As in Section~\ref{sec:CA}, let $\Gamma$ be a bipartite graph isoradially embedded is a flat
surface $\Sigma$ with cone type singularities $S\subset B\cup V(\G^*)$.

\begin{lemma}\label{lemma:flat}
Given a nowhere vanishing vector field $V$ along $W$, let $\omega_V\in C^1(\Gamma;S^1)$ be the cochain defined by $\omega_V(e)=\exp(i\vartheta_V(w,b))$
as illustrated in Figure~\ref{fig:theta}. Then, the equivalence class of $\omega_V$ does not depend on $V$.
Furthermore, its Kasteleyn curvature is given by
\[
c_{\omega_V}(f)=-\exp(i\theta_f/2),
\]
where $\theta_f$ denotes the angle of the conical singularity in the face $f$.
\end{lemma}
\begin{proof}
The first statement is obvious, so let us fix a nowhere vanishing vector field $V\in\X(W)$ and consider the associated
cochain $\omega_V$. Given a face $f$ of $\G\subset\Sigma$, let $w_1,b_1,w_2,b_2,\dots,w_m,b_m$ denote the vertices
in $\partial f$ cyclically ordered. Then, the Kasteleyn curvature of $\omega_V$ at the face $f$ is given by
\begin{eqnarray*}
c_{\omega_V}(f)&=&(-1)^{\frac{\left|\partial f\right|}{2}+1}\omega_V(\partial f)\\
&=&(-1)^{m+1}\omega_V(w_1,b_1)\omega_V(b_1,w_2)\omega_V(w_2,b_2)\cdots\omega_V(b_m,w_1)\\
&=&-(-1)^m\frac{\omega_V(w_1,b_1)\omega_V(w_2,b_2)\cdots\omega_V(w_m,b_m)\phantom{_{-1}}}
{\omega_V(w_1,b_m)\omega_V(w_2,b_1)\cdots\omega_V(w_m,b_{m-1})}\\
&=&-\exp\big(i\sum_{j=1}^m(\pi-\alpha_{w_j}(\partial f))\big),
\end{eqnarray*}
where $\alpha_{w_j}(\partial f)=\vartheta_V(w_j,b_{j-1})-\vartheta_V(w_j,b_j)$ and $b_0=b_m$. This angle $\alpha_w(\partial f)$ is simply the angle made by
the oriented curve $\partial f$ at the vertex $w$, as illustrated in Figure~\ref{fig:theta2}. An easy application of Gauss-Bonnet shows that the angle $\theta_f$
of the conical singularity $x_f$ in $f$ is equal to $\sum_{v\in\partial f}(\pi-\alpha_v(\partial f))$. Hence, it remains to check that
\[
\sum_{b\in B\cap\partial f}\alpha_b(\partial f)-\sum_{w\in W\cap\partial f}\alpha_w(\partial f)=0.
\]
This is where the isoradiality comes into play. By definition, there is a local isometry from the pointed face
$f\setminus\{x_f\}$ to the pointed plane $\C^*$ such that all vertices in $\partial f$ are mapped to a circle in
$\C^*$. Now, observe that the alternating sum of angles displayed above does not change if one moves a vertex along
the circle keeping all other vertices fixed. Since the equality above holds when all angles are equal (to $\pi-\theta_f/2m$),
this concludes the proof.
\end{proof}

\begin{figure}[htb]
\labellist\small\hair 2.5pt
\pinlabel {$v$} at 168 18
\pinlabel {$\alpha_v(C)$} at 175 160
\pinlabel {$C$} at 420 190
\endlabellist
\centerline{\psfig{file=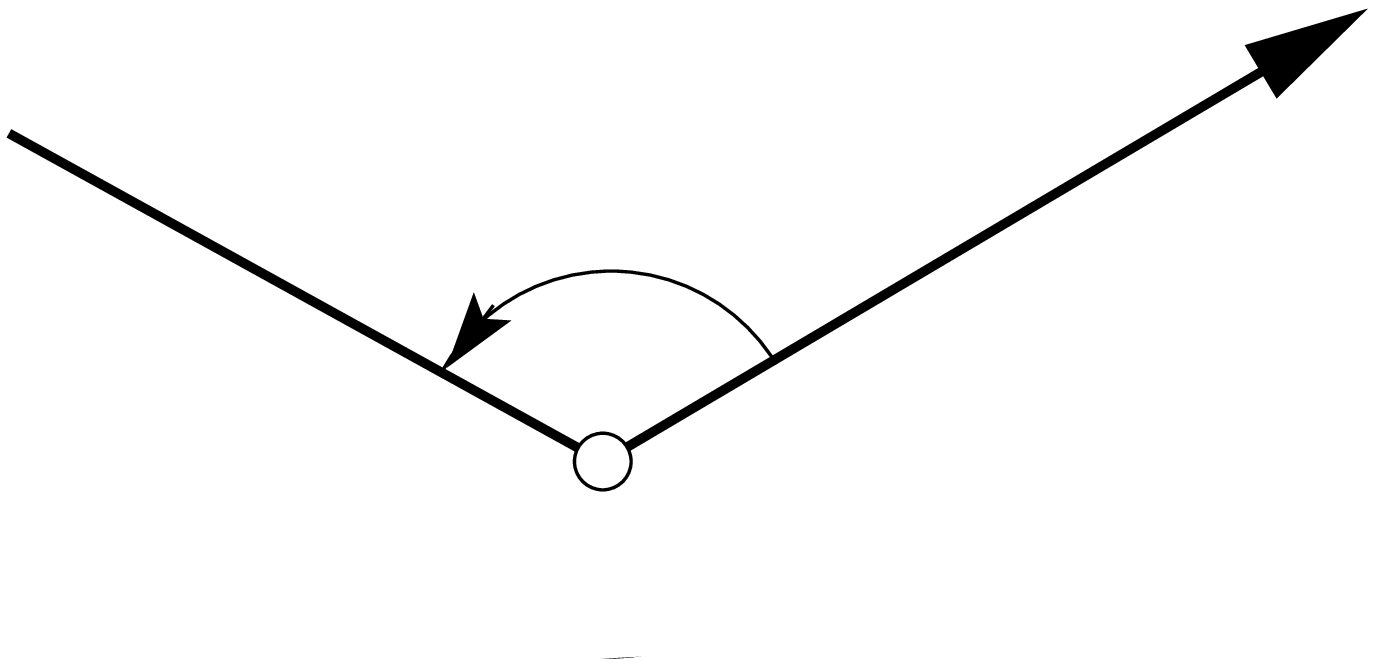,height=1.8cm}}
\caption{The angle made by the oriented curve $C$ at the vertex $v$.}
\label{fig:theta2}
\end{figure}

Let us now state the main result of this section.

\begin{theorem}\label{thm:main}
Let $\Sigma$ be a compact oriented flat surface of genus $g$ with conical singularities supported at $S$ and cone angles multiples of $2\pi$.
Fix a graph $\Gamma$ with bipartite structure $V(\Gamma)=B\sqcup W$, isoradially embedded in $\Sigma$ so that $S\subset B\cup V(\Gamma^*)$.
For an edge $e$ of $\Gamma$, let $\nu(e)$ denote the length of the dual edge.
Finally, let $D_\lambda\colon\C^B\to\C^W$ denote the discrete Dirac operator associated to the discrete spin structure $\lambda$.

There exist $2^{2g}$ non-equivalent discrete spin structures such that the corresponding discrete Dirac operators $\{D_\lambda\}_\lambda$
give $2^{2g}$ non-equivalent Kasteleyn matrices of the weighted graph $(\Gamma,\nu)$, if and only if the following conditions hold:
\begin{romanlist}
\item{each conical singularity in $V(\G^*)$ has angle an odd multiple of $2\pi$;}
\item{for some (or equivalently, for any) choice of oriented simple closed curves $\{C_j\}$ in $\Gamma$ representing
a basis of $H_1(\Sigma;\Z)$,
\[
\sum_{b\in B\cap C_j}\alpha_b(C_j)-\sum_{w\in W\cap C_j}\alpha_w(C_j)
\]
is a multiple of $2\pi$ for all $j$, where $\alpha_v(C)$ denotes the angle made by the oriented curve $C$
at the vertex $v$ as illustrated in Figure~\ref{fig:theta2}.}
\end{romanlist}
\end{theorem}
\begin{proof}
Fix a discrete spin structure $\lambda\in Z^1(X;S^1)$, a normalized vector field $V\in\X(W)$, and let $D_\lambda\colon\C^B\to\C^W$ be the corresponding discrete Dirac operator.
By definition, the coefficient of $D_\lambda$ corresponding to vertices $w\in W$ and $b\in B$ is equal to
\[
D_\lambda(w,b)=
\begin{cases}
\nu(e)\,\omega_V(e)\lambda(e) & \text{if $w$ and $b$ are joined by an edge $e$;} \\
0 & \text{if $w$ and $b$ are not adjacent,}
\end{cases}
\]
with $\omega_V$ as in Lemma~\ref{lemma:flat}. In other words, $D_\lambda$ is the adjacency matrix
of the weighted bipartite graph $(\Gamma,\nu)$, twisted by the cochain $\omega_\lambda:=\omega_V\lambda\in C^1(\Gamma;S^1)$.
The goal is now to check that $\omega_\lambda$ is gauge equivalent to a Kasteleyn orientation (that is, to a $\pm 1$-valued
flat cochain) if and only if Conditions $(i)$ and $(ii)$ hold. This clearly implies the theorem, as
non-equivalent discrete spin structures yield non-equivalent Kasteleyn orientations.

By Lemma~\ref{lemma:flat}, $\omega_V$ is flat if and only if Condition $(i)$ is satisfied.
Since $\lambda$ is a cocycle,
$\delta\lambda=1$ and the same statement holds for $\omega_\lambda$. By Proposition~\ref{prop:torsor}, the set of equivalence
classes of such $S^1$-valued flat cochains is an $H^1(\Sigma;S^1)$-torsor. Therefore, $\omega_\lambda$ is equivalent to
a Kasteleyn orientation if and only if, for any Kasteleyn orientation $\omega_0$, the cocycle
$\phi:=\omega_0^{-1}\omega_\lambda$ represents a class in $H^1(\Sigma;\{\pm 1\})=\mathrm{Hom}(H_1(\Sigma;\Z),\{\pm 1\})$.
This holds if and only if $\phi(C)\in\{\pm 1\}$ for any 1-cycle $C$ in $\Gamma$, or equivalently, for any
1-cycle in $\Gamma$ part of a collection representing a basis of $H_1(\Sigma;\Z)$. Since $\omega_0^2=1$, this
translates into the equalities
\[
1=\phi(C)^2=\omega_\lambda(C)^2=\omega_V(C)^2\lambda(C)^{2}=\omega_V(C)^2\kappa(C).
\]
As mentioned in Section~\ref{sub:DG'}, $\kappa(C)$ is the inverse of the holonomy along the 1-cycle $C$. Therefore,
\[
\kappa(C)=\mathit{hol}(C)^{-1}=\exp\Big(-i\sum_{v\in V(\Gamma)\cap C}(\pi-\alpha_v(C))\Big)=\exp\Big(i\sum_{v\in V(\Gamma)\cap C}\alpha_v(C)\Big),
\]
since $C$ is of even length.
Furthermore, the definition of $\omega_V$ implies that $\omega_V(C)=\exp(-i\sum_{w\in W\cap C}\alpha_w(C))$, as in the proof of Lemma~\ref{lemma:flat}.
This yields the equation
\[
1=\exp\Big(i\sum_{b\in B\cap C}\alpha_b(C)-i\sum_{w\in W\cap C}\alpha_w(C)\Big),
\]
obviously equivalent to Condition $(ii)$.
\end{proof}

Consider $\Gamma\subset\Sigma$ as above, and satisfying both conditions of Theorem~\ref{thm:main}. Let
$\mathcal{B}=\{\alpha_j\}$ be a set of simple closed curves on $\Sigma$, transverse to $\Gamma$, whose classes form a
basis of $H_1(\Sigma;\Z)$. For each $\alpha_j\in\mathcal{B}$, let $C_j$ denote the oriented 1-cycle in $\Gamma$
having $\alpha_j$ to its immediate left, and meeting every vertex of $\G$ adjacent to $\alpha_j$ on this side.
By the conditions of Theorem~\ref{thm:main}, any discrete spin structure $\lambda$ satisfies the equations
\[
\lambda(C_j)=\exp\Big(i\sum_{w\in W\cap C_j}\alpha_w(C_j)\Big)(-1)^{\frac{|C_j|}{2}+1}\tag{$\star$}
\]
up to a sign. Let us pick the discrete spin structure $\lambda_0$ such that the equality above holds for all $j$.
For any $\eps=(\eps_1,\dots,\eps_{2g})\in\Z_2^{2g}$, let $\lambda_\eps$ denote the discrete spin structure obtained
from $\lambda_0$ as follows:
\[
\lambda_\eps(e)=(-1)^{\sum_j\eps_j(e\cdot\alpha_j)}\lambda_0(e).
\]

\begin{theorem}\label{thm:Pf-D}
If $\G\subset\Si$ satisfies the conditions of Theorem~\ref{thm:main}, then the partition function for the dimer model on $(\Gamma,\nu)$ is given by
\[
Z(\Gamma,\nu)=\frac{1}{2^g}\Big|\sum_{\eps\in\Z_2^{2g}}(-1)^{\sum_{i<j}\eps_i\eps_j\alpha_i\cdot\alpha_j}\det(D_{\lambda_\eps})\Big|.
\]
\end{theorem}
\begin{proof}
By Condition $(i)$, the $S^1$-valued cochain $\lambda_0\omega_V$ is Kasteleyn flat.
By Condition $(ii)$, it is gauge equivalent to a $\{\pm 1\}$-valued cocycle $\tau$.
Finally, Equation ($\star$) is equivalent to $\tau(C_j)=(-1)^{|C_j|/2+1}$.
The theorem now follows from Corollary~\ref{cor:Pf} for $G=S^1\subset\mathbb{F}^*=\C^*$.
\end{proof}

\begin{remark}
More generally, let us assume that $\G\subset\Si$ only satisfies the first condition of Theorem~\ref{thm:main}, and let $\lambda$ be any discrete spin structure. Then,
Theorem~\ref{thm:Pf-K} gives the equality
\[
\sum_{\alpha\in H_1(\Sigma;\Z)}\varphi(\alpha)Z^\mathcal{B}_\alpha(\Gamma,\nu)=
\frac{1}{2^g}\Big|\sum_{\eps\in\Z_2^{2g}}(-1)^{\sum_{i<j}\eps_i\eps_j\alpha_i\cdot\alpha_j}\det(D_{\lambda_\eps})\Big|,
\]
where $\varphi\in H^1(\Si;\Z)$ is such that $\varphi\cdot[\tau]=[\lambda\omega_V]$.
\end{remark}

Spin structures on a closed orientable surface $\Si$ can be identified with {\bf quadratic forms}, that is, with $\Z_2$-valued maps on
$H_1(\Si;\Z_2)$ such that $q(x+y)=q(x)+q(y)+x\cdot y$ for all $x,y$ in $H_1(\Si;\Z_2)$. More precisely, Johnson~\cite{Joh} gave an explicit
$H^1(\Si;\Z_2)$-equivariant bijection $\S(\Si)\stackrel{\varphi}{\to}\Q(\Si)$ between the corresponding affine $H^1(\Si;\Z_2)$-spaces.
The {\bf Arf invariant} of a spin structure is then defined as the Arf invariant of the corresponding quadratic form $q$, that is, the mod 2 integer $\Arf(q)\in\Z_2$ such that
\[
(-1)^{\Arf(q)}=\frac{1}{\sqrt{|H_1(\Si;\Z_2)|}}\sum_{x\in H_1(\Si;\Z_2)} (-1)^{q(x)}.
\]
If all cone angles of $\Si$ are odd multiples of $2\pi$, then there exists a canonical equivariant bijection $\S(X)\to\S(\Si)$ (recall Proposition~\ref{prop:spin}).
In such a case, it makes sense to talk about the Arf invariant $\Arf(\lambda)$ of a discrete spin structure $\lambda\in\S(X)$.

As above, let  $\{\alpha_j\}$ be a set of simple closed curves on $\Si$, transverse to $\G$, defining a basis of $H_1(\Si;\Z)$, and let $C_j$ denote the oriented cycle in $\G$
having $\alpha_j$ to its immediate left. By Condition $(ii)$, the number
\[
q_0(\alpha_j)=\frac{1}{2\pi}\Big(\sum_{w\in W\cap C_j}\alpha_w(C_j)-\sum_{b\in B\cap C_j}\alpha_b(C_j)\Big)
\]
is an integer. Furthermore, one easily checks that its parity changes each time $\alpha_j$ moves across one vertex. Therefore, the $\alpha_j$'s can be chosen so that all
$q_0(\alpha_j)'s$ are even.

This leads to the following version of the Pfaffian formula, assuming the notations preceding Theorem~\ref{thm:Pf-D}.

\begin{theorem}\label{thm:Pf-Arf}
Let $\G\subset\Si$ be as in the statement of Theorem~\ref{thm:main}, with all cone angles of $\Si$ odd multiples of $2\pi$.
Then, the partition function for the dimer model on $(\Gamma,\nu)$ is given by
\[
Z(\Gamma,\nu)=\frac{1}{2^g}\Big|\sum_{\eps\in\Z_2^{2g}}(-1)^{\Arf(\lambda_\eps)}\det(D_{\lambda_\eps})\Big|.
\]
\end{theorem}
\begin{proof}
We saw in the proof of Proposition~\ref{prop:spin} that the quadratic form $q_\lambda\colon H_1(\Si;\Z_2)\to\Z_2$ corresponding to a class
$[\lambda]\in\S(X)$ via the equivariant bijection $\S(X)\to\S(\Si)\stackrel{\varphi}{\to}\Q(\Si)$ is determined by the following condition:
for any oriented simple closed curve $C\subset\G$,
\[
q_\lambda([C])=1+\frac{|C|}{2}+\frac{1}{2\pi}\Big(\sum_{e\subset C} 2\beta_\lambda(e)-\sum_{v\in C}\alpha_v(C)\Big),
\]
where $0\le\beta_\lambda(e)<2\pi$ is such that $\lambda(e)=\exp(i\beta_\lambda(e))$.
In particular, if $\G\subset\Si$ satisfies Condition $(ii)$ of Theorem~\ref{thm:main}, then the discrete spin structure $\lambda_0$ defined by Equation $(\star)$ corresponds to
the quadratic form $q_0$ determined by the equalities
\[
q_0(\alpha_j)=q_0([C_j])=\frac{1}{2\pi}\Big(\sum_{w\in W\cap C_j}\alpha_w(C_j)-\sum_{b\in B\cap C_j}\alpha_b(C_j)\Big).
\]
By construction, $\lambda_\eps$ is obtained from $\lambda_0$ by action of the Poincar\'e dual to $\Delta_\eps=\sum_j\eps_j\alpha_j\in H_1(\Si;\Z_2)$.
Therefore, by \cite[Lemma 1]{CR1},
\[
\Arf(\lambda_\eps)+\Arf(\lambda_0)=q_0(\Delta_\eps)=\sum_j\eps_j q_0(\alpha_j)+\sum_{i<j}\eps_i\eps_j\alpha_i\cdot\alpha_j.
\]
As we have chosen the $\alpha_j$'s so that all $q_0(\alpha_j)$'s are even, the formula now follows from Theorem~\ref{thm:Pf-D}.
\end{proof}

\begin{remark}
As mentioned in Remark~\ref{rem:W}, it is not necessary in the present section to assume that the sets $S$ and $W$ are disjoint.
All the results of this section still hold in the slightly more general setting where $S$ is contained in $V(\G)\cup V(\G^*)$.
\end{remark}

\subsection{Examples}
\label{sub:ex}

We conclude this article with a discussion of several special cases, and the following result: the Dirac operators on any closed Riemann surface of positive genus
can be approximated by Kasteleyn matrices.

\noindent{\bf The planar case.} Let $\G$ be a planar isoradial bipartite graph whose associated rhombic lattice tiles a simply-connected domain $\Si$ of the plane.
In this case, the unique spin structure on $\Si$ being trivial, the associated discrete Dirac operator $D\colon\C^B\to\C^W$ is simply given by
\[
(D\psi)(w)=\sum_{b\sim w}\nu(w,b)e^{i\vartheta_V(w,b)}\,\psi(b),
\]
where the angle $\vartheta_V(w,b)$ can be measured with respect to a constant vector field $V$. With the notations of Figure~\ref{fig:star}, this yields the equality
\[
(D\psi)(w)=i\,\sum_{j=1}^m(x_{j-1}-x_{j})f(b_j),
\]
which is exactly the discrete Dirac operator introduced by Kenyon \cite{Ken} in this special case. The conditions of Theorem~\ref{thm:main} being trivially satisfied,
$D$ is (conjugate to) a Kasteleyn matrix for the dimer model on $(\G,\nu)$, and the associated partition function is given by
\[
Z(\G,\nu)=|\det(D)|.
\]
Thus, in the planar case, we recover Theorem 10.1 of \cite{Ken}.

\smallskip

\noindent{\bf The genus one case.} Let $\widetilde\G$ be a planar isoradial bipartite graph, invariant under the action of the lattice $\Lambda=\Z 1\oplus\Z\tau\subset\C$
for some $\tau\in\mathbb{H}$. Fix a quadrilateral fundamental domain $P\subset\C$ for this action with $\G$ intersecting $\partial P$ transversally, and let
$\G\subset\Sigma=\C/\Lambda$ be the corresponding toric graph. One of the spin structures on $\Si$ being trivial, the associated discrete Dirac operator 
$D\colon\C^B\to\C^W$ is again given by
\[
(D\psi)(w)=\sum_{b\sim w}\nu(w,b)e^{i\vartheta_V(w,b)}\,\psi(b),
\]
where $V$ can be chosen to be a constant vector field.
The three other discrete Dirac operators are obtained from $D$ by multiplying the corresponding coefficient by $-1$ whenever an edge crosses the horizontal boundary
components of $P$ (this gives $D_{1,0}$), the vertical ones ($D_{0,1}$), or any boundary component ($D_{1,1}$). These 4 matrices are Kasteleyn matrices if and only if
$\G$ satisfies Condition $(ii)$ in Theorem~\ref{thm:main}, in which case $Z(\G,\nu)$ can be written as an alternating sum of the determinant of these matrices.

\begin{example}
Consider the graph illustrated below. We have inserted next to each edge the corresponding coefficient of the matrix $D$.
(The graph is normalized so that the sides of $P$ have length 3, $V$ is chosen to be the vertical upward direction, and $\omega$ stands for $\exp(2\pi i/3)$).

\begin{figure}[htb]
\labellist\small\hair 2.5pt
\pinlabel {$1$} at 600 720
\pinlabel {$1$} at 890 720
\pinlabel {$1$} at 1180 720
\pinlabel {$1$} at 460 500
\pinlabel {$1$} at 750 500
\pinlabel {$1$} at 1040 500
\pinlabel {$1$} at 320 250
\pinlabel {$1$} at 610 250
\pinlabel {$1$} at 900 250
\pinlabel {$\omega$} at 485 655
\pinlabel {$\omega$} at 775 655
\pinlabel {$\omega$} at 1065 655
\pinlabel {$\bar\omega$} at 670 665
\pinlabel {$\bar\omega$} at 960 665
\pinlabel {$\bar\omega$} at 1220 665
\pinlabel {$\omega$} at 340 400
\pinlabel {$\omega$} at 630 400
\pinlabel {$\omega$} at 920 400
\pinlabel {$\bar\omega$} at 520 410
\pinlabel {$\bar\omega$} at 810 410
\pinlabel {$\bar\omega$} at 1075 415
\pinlabel {$\omega$} at 190 150
\pinlabel {$\omega$} at 480 150
\pinlabel {$\omega$} at 770 150
\pinlabel {$\bar\omega$} at 380 155
\pinlabel {$\bar\omega$} at 670 155
\pinlabel {$\bar\omega$} at 935 165
\endlabellist
\centerline{\psfig{file=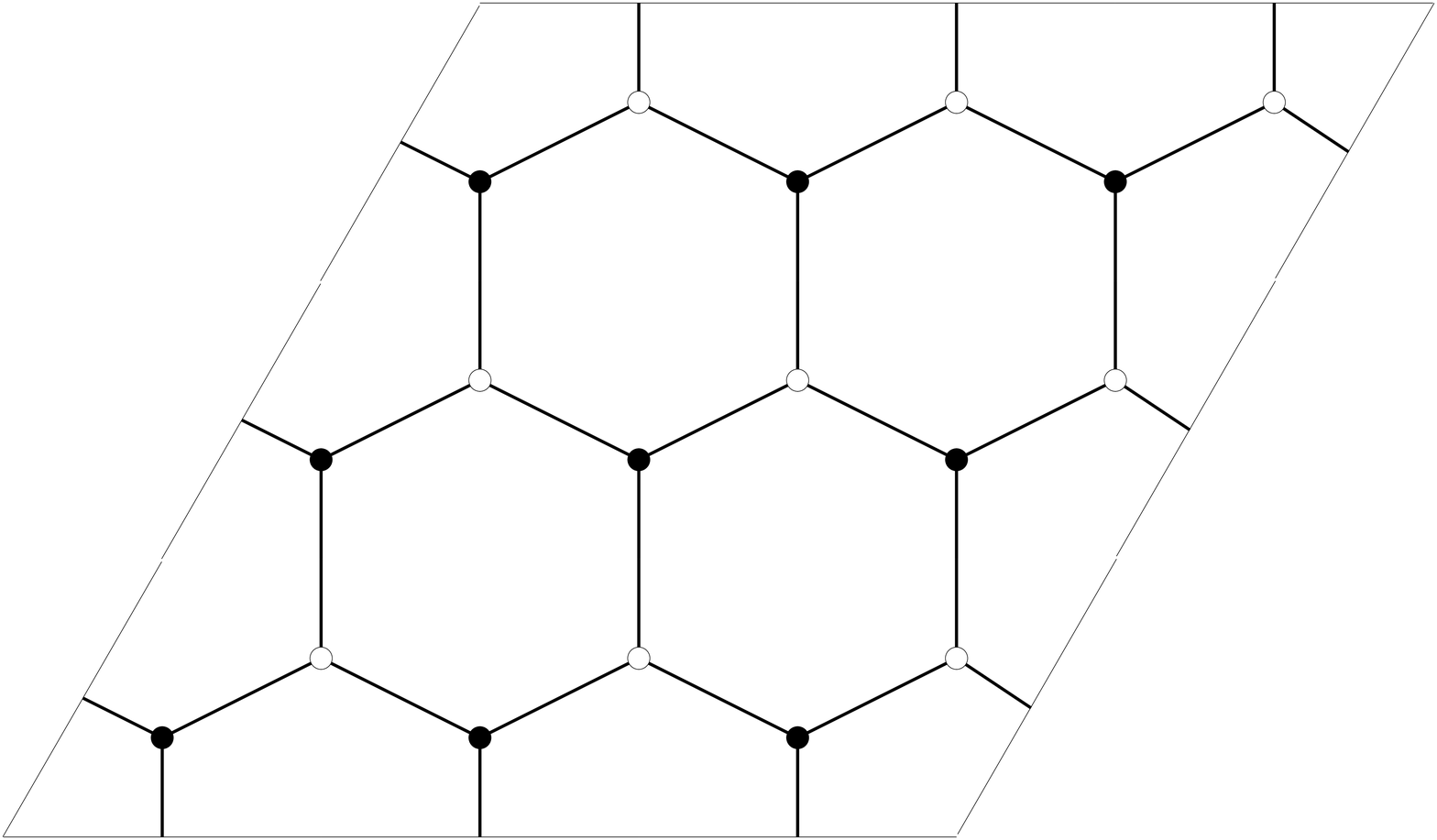,height=5cm}}
\end{figure}

In this example, Condition $(ii)$ in Theorem~\ref{thm:main} is satisfied.
Furthermore, one easily checks that the trivial discrete spin structure satisfies Equation $(\star)$,
where $\alpha_1$ and $\alpha_2$ are chosen to be the sides of $P$.
Therefore, Theorem~\ref{thm:Pf-D} gives the equality
\begin{eqnarray*}
\#\mathcal{M}(\G)&=&Z(\G,1)=\frac{1}{2}\left|\det(D)+\det(D_{1,0})+\det(D_{0,1})-\det(D_{1,1})\right|\\
&=& \frac{1}{2}\left|0+28+28-(-28)\right|=42.
\end{eqnarray*}
\end{example}

\smallskip

\noindent{\bf An example of genus 2.} Consider the flat surface $\Si$ of genus 2 given by an octagon, where all pairs of opposite sides are identified.
Embed a square lattice $\G$ in $\Si$ as illustrated below.

\begin{figure}[h]
\centerline{\psfig{file=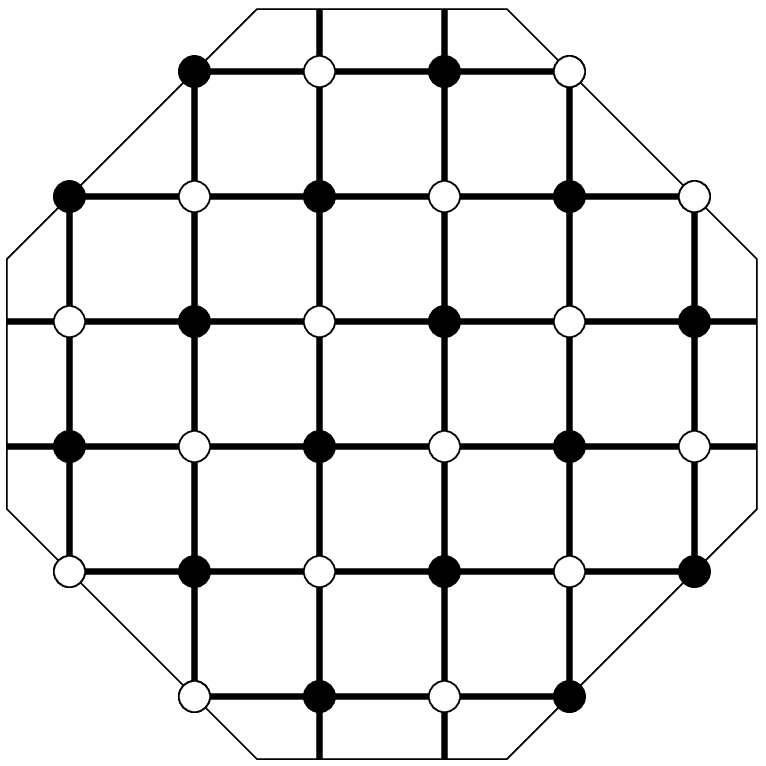,height=5cm}}
\end{figure}

The flat surface $\Si$ has a single singularity, which lies in $V(\G^*)$, and has angle $6\pi$. Therefore, this example satisfies the first condition of Theorem~\ref{thm:main}.
One easily checks that it also satifies the second condition, so the 16 discrete Dirac operators are Kasteleyn matrices. Since $\Si$ has trivial holonomy, one discrete spin
structure can be chosen to be trivial. Furthermore, one can fix a constant direction $V$ on $\Si$ (say, to the right). The corresponding discrete Dirac operator
$D\colon\C^B\to\C^W$ is simply given by
\[
(D\psi)(w)=(\psi(b_1)-\psi(b_3))+i(\psi(b_2)-\psi(b_4)),
\]
where $b_1$ is the black vertex to the right of $w$, $b_2$ above, $b_3$ to the left, and $b_4$ below. Using the procedure described before Theorem~\ref{thm:Pf-D}, it is now
a trivial matter to write the number of dimer coverings of $\G$ as some alternating sum of determinants of these 16 discrete Dirac operators.

\smallskip

The example above only discretizes the Dirac operators on one specific Riemann surface of genus 2. Can one find examples for any closed Riemann surface?
Obviously not for the Riemann sphere, as all cone angles are assumed to be positive multiples of $2\pi$. However, this turns out to be the only exception,
as demonstrated by the following theorem.

\begin{theorem}\label{thm:real}
For any closed Riemann surface of positive genus, there exist a flat surface $\Si$ with cone type singularities inducing this complex structure, and an isoradially
embedded bipartite graph $\G\subset\Si$, with arbitrarily small radius, satisfying all the hypothesis and conditions of Theorem~\ref{thm:main}.
\end{theorem}
\begin{proof}
The building block of our construction will be the rhombus consisting of two equilateral triangles glued along one of their sides. Given positive integers $n$ and $m$,
let $R(m,n)$ denote $m$ rows of $n$ such rhombi stacked in the following way. (This picture represents $R(2,8)$.)

\begin{figure}[h]
\centerline{\psfig{file=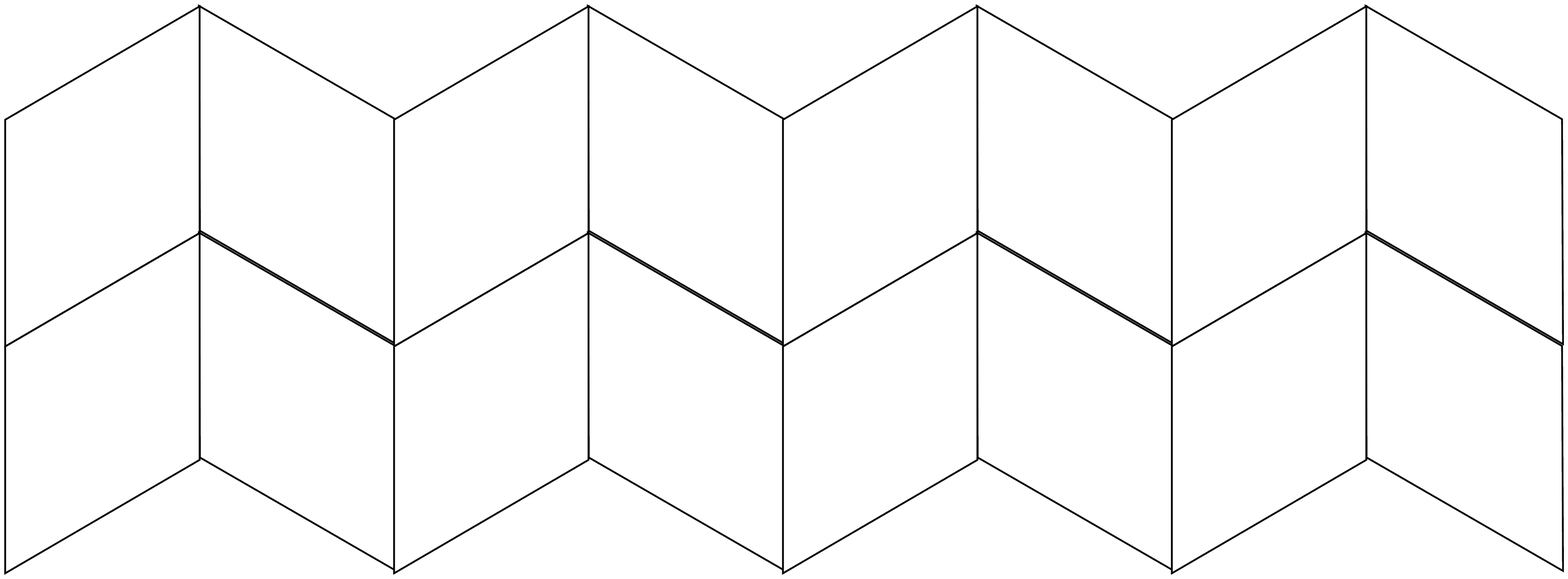,height=1.5cm}}
\end{figure}

\noindent
Let $\G$ be the associated bipartite isoradial planar graph, where the bottom left corner of $R(m,n)$ is a black vertex of $\G$. If $n$ is even, the identification of the two
vertical sides of $R(m,n)$ will preserve the bipartite structure of $\G$. If $m$ is even, one can also identify the horizontal sides, possibly with a shift. This allows to
realize any torus $\C/\Z+\Z\tau$ with $\tau$ in some dense subset of $\mathbb{H}$. To obtain all tori, continuously deform one or two rows of rhombi as illustrated below.

\begin{figure}[h]
\centerline{\psfig{file=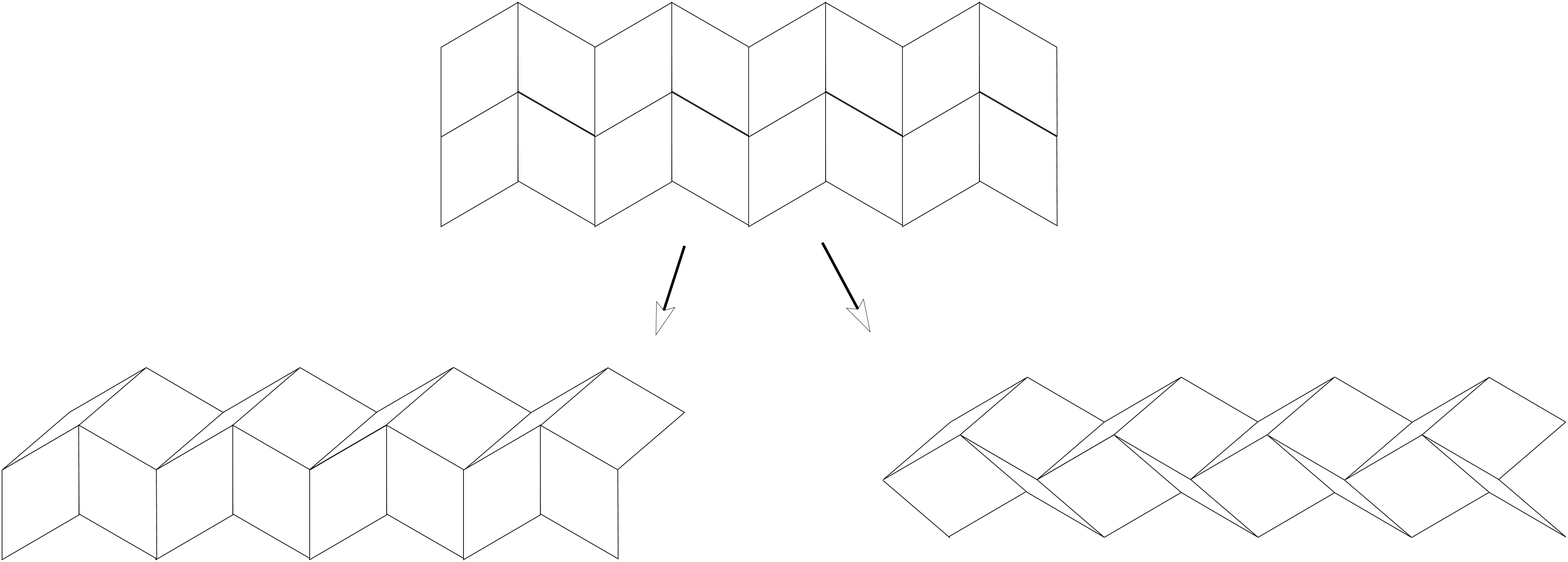,height=4.3cm}}
\end{figure}

\noindent The deformation of two rows changes the imaginary part of $\tau$, while the deformation of a single row changes both the imaginary and the real parts.
Therefore, a suitable
combination of these transformations allows to construct all tori. These examples are flat tori with no singularity, so they trivially satisfy the first condition in
Theorem~\ref{thm:main}. Furthermore, one easily checks that Condition $(ii)$ is also satisfied, provided $n$ and $m$ are divisible by 6. (Note that the deformations
above do not affect these conditions.)

Let us now consider a fixed positive even integer $n$. Given three positive integers $m_1,m_2,m_3$,
glue the corresponding rectangles $R(m_1,n)$, $R(m_2,n)$ and $R(m_3,n)$ along their bottom side to an equilateral triangle, itself tiled by rhombi, as illustrated below.

\begin{figure}[h]
\centerline{\psfig{file=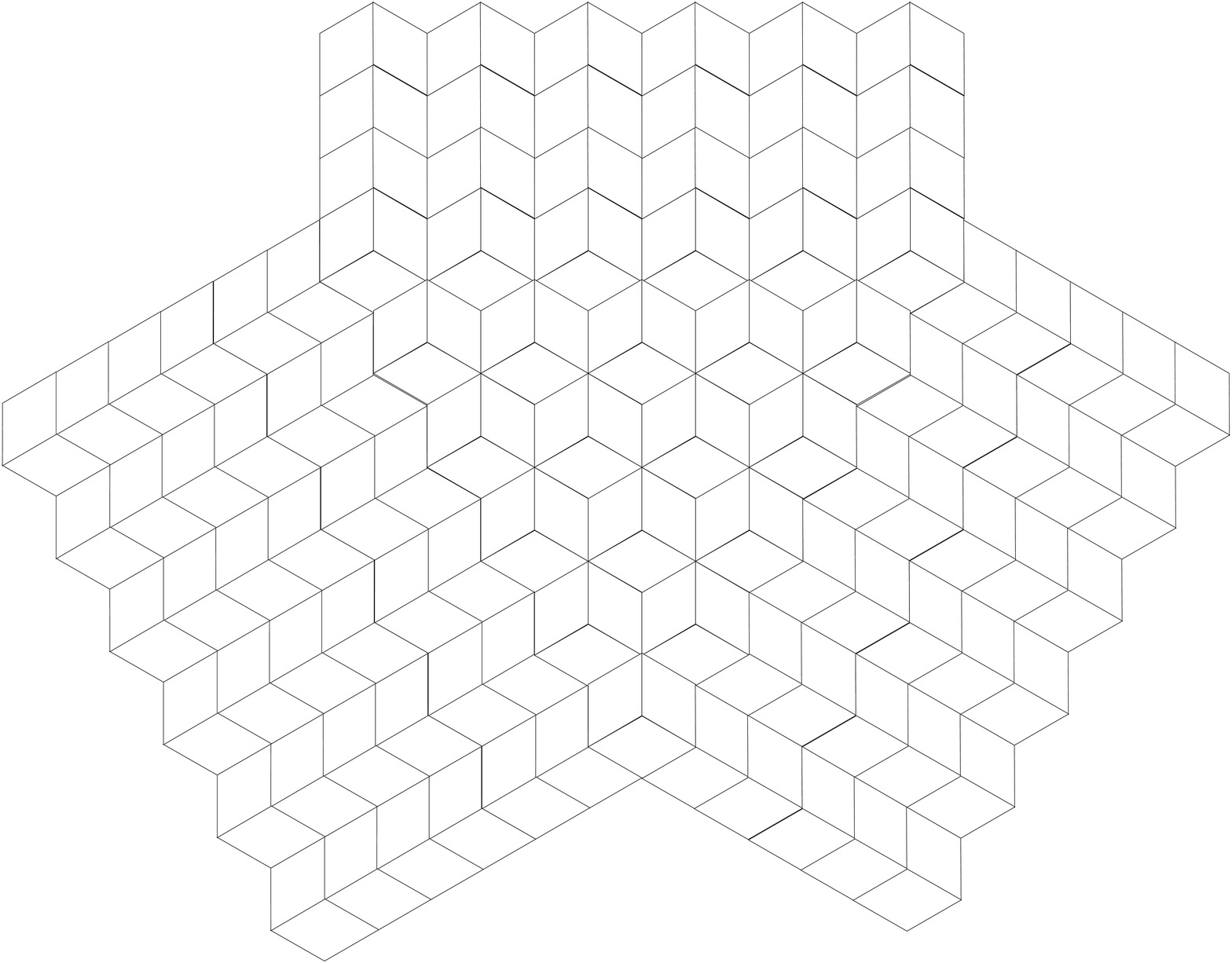,height=6cm}}
\end{figure}

\noindent
Identifying the opposite remaining sides of each rectangle yields a flat pair of pants with a single singularity of angle $4\pi$.
By varying the $m_j$'s and using the deformation along two rows described above, one can realize any complex structure on the pair of pants. Finally,
gluing $2g-2$ such pairs of pants along their boundaries (with a possible shift and a possible deformation yielding a twist), one can realize any Riemann surface of
genus $g\ge 2$.

To each rhombus, associate the portion of a bipartite graph $\G$ given by $\begin{array}{c}\includegraphics[width=0.5in]{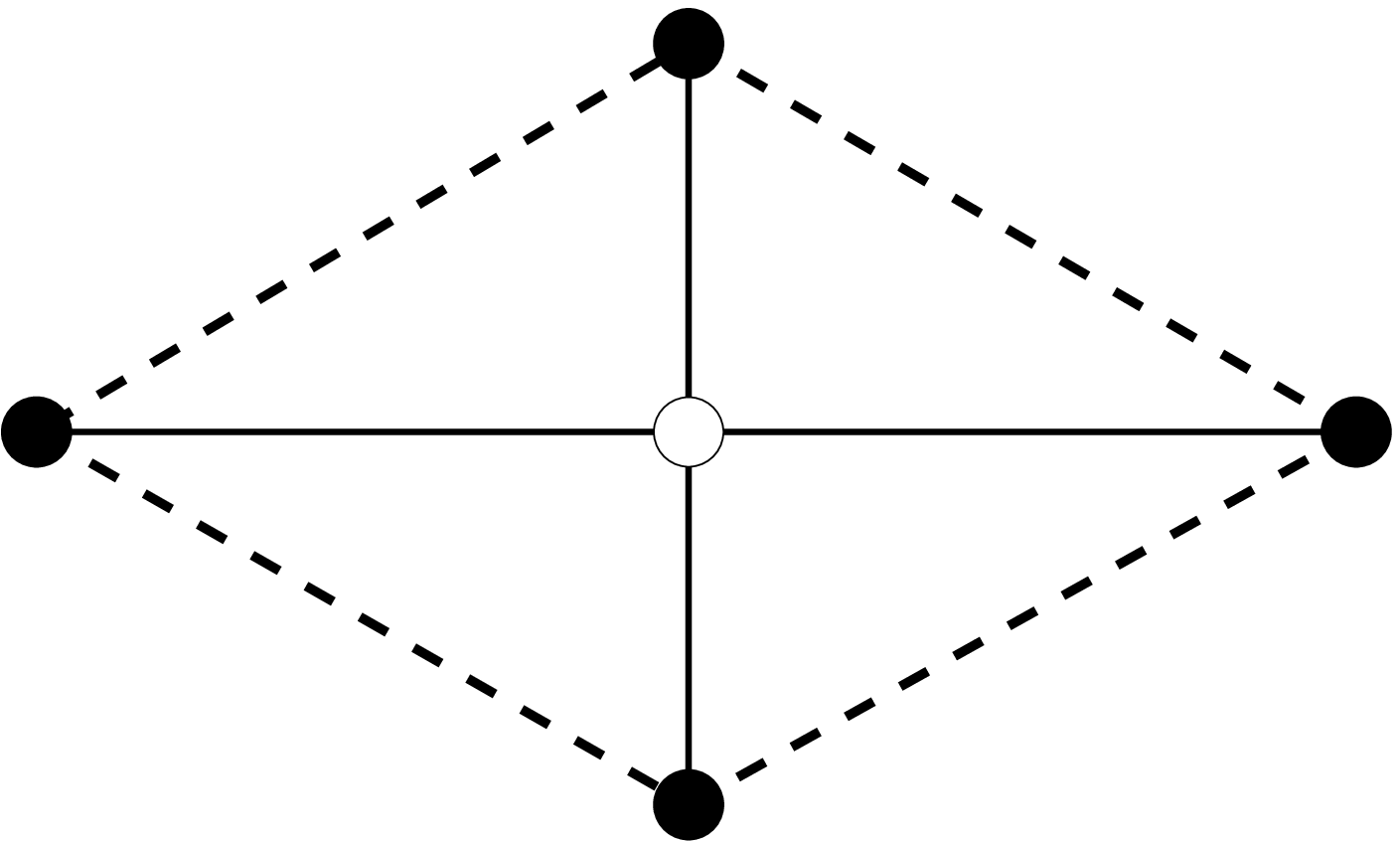}\end{array}$. (This is just to avoid
cumbersome considerations about gluing bipartite structures.) The singularities of angle $4\pi$ are located at black vertices of $\G$, so Condition $(i)$ is satisfied.
One easily checks that Condition $(ii)$ is always satisfied for cycles coming from boundary components of the pairs of pants. Finally, by chosing wisely the parity of the
$m_i$'s, one can ensure that Condition $(ii)$ also holds for the cycles passing through several pairs of pants.
\end{proof}

\bibliographystyle{amsplain}
\bibliography{Dirac}

\end{document}